\documentclass[12pt]{article}

\usepackage{a4wide, amsmath,amsthm,amsfonts,amscd,amssymb,eucal,bbm,mathrsfs, enumerate}

\usepackage{graphicx,tikz}
\usepackage[pdfborder={0 0 0}]{hyperref}
\usetikzlibrary{decorations.pathreplacing}
\usetikzlibrary{decorations.pathmorphing}
\usetikzlibrary{arrows}

\def\RR{{\mathbb R}}
\def\CC{{\mathbb C}}

\def\ZZ{{\mathbb Z}}

\def\A{{\mathcal A}}
\def\B{{\mathcal B}}

\def\D{{\mathcal D}}
\def\E{{\mathcal E}}
\def\F{{\mathcal F}}
\def\H{{\mathcal H}}

\def\K{{\mathcal K}}
\def\M{{\mathcal M}}
\def\N{{\mathcal N}}

\def\P{{\mathcal P}}
\def\R{{\mathcal R}}
\def\S{{\mathcal S}}

\def\U{{\mathcal U}}

\def\a{\alpha}
\def\b{\beta}
\def\d{\delta}

\def\f{\varphi}
\def\g{\gamma}

\def\k{\kappa}

\def\L{{\mathrm L}}

\def\Q{\Omega}

\def\R{{\mathrm R}}
\def\s{\sigma}

\def\Ad{{\hbox{\rm Ad\,}}}

\def\dim{{\hbox{dim}\,}}

\def\sp{{\rm sp}\,}

\def\1{{\mathbbm 1}}

\def\uone{{\rm U(1)}}
\def\u1net{{\A^{(0)}}}

\def\diff{{\rm Diff}}

\def\diffs1{\diff(S^1)}
\def\mob{{\rm M\ddot{o}b}}
\def\mob2{{\rm M\ddot{o}b}^{(2)}}

\def\supp{{\rm supp}}

\def\psl2r{{\rm PSL}(2,\RR)}
\def\sl2r{{\rm SL}(2,\RR)}
\def\su11{{\rm SU}(1,1)}
\def\2dmob{{\overline{\psl2r}\times\overline{\psl2r}}}
\def\<{\langle}
\def\>{\rangle}

\DeclareMathOperator{\Mob}{M\ddot ob}
 
\newcommand{\dd}{\,\mathrm{d}} 
 
\newcommand{\ee}{\mathrm{e}} 
\newcommand{\ima}{\mathrm{i}}

\newcommand{\slot}{\,\cdot\,} 
\newcommand{\FerC}{ \mathrm{Fer}_\CC}

\newcommand{\Mr}{\mathcal{M}_{\mathrm{r}}} 
\newcommand{\Hr}{\mathcal{H}_{\mathrm{r}}}

\newcommand{\Trr}{T_{\mathrm{r}}} 
\newcommand{\Omr}{\Omega_{\mathrm{r}}} 

\newcommand{\Mc}{\mathcal{M}_{\mathrm{c}}} 
 
\newcommand{\Tc}{T_{\mathrm{c}}} 
\newcommand{\Qc}{Q_{\mathrm{c}}} 
\newcommand{\Omc}{\Omega_{\mathrm{c}}} 

\newcommand{\Mtc}{\widetilde{\mathcal{M}}_{\mathrm{c}, \k}} 
 
\newcommand{\Ttc}{\widetilde{T}_{\mathrm{c}}} 
\newcommand{\Omtc}{\widetilde{\Omega}_{\mathrm{c}}} 
\usepackage{ushort}
\newcommand{\und}{\ushort}

\newtheorem{theorem}{Theorem}[section]
\newtheorem{definition}[theorem]{Definition}
\newtheorem{corollary}[theorem]{Corollary}
\newtheorem{proposition}[theorem]{Proposition}
\newtheorem{lemma}[theorem]{Lemma}
\theoremstyle{remark}
\newtheorem{remark}[theorem]{Remark}

\title{Integrable QFT and Longo-Witten endomorphisms}
\date{}
\author{{\bf Marcel Bischoff}%
\footnote{Supported by the German Research Foundation (Deutsche Forschungsgemeinschaft (DFG))
by the DFG Research Training Group 1493 ``Mathematical Structures in Modern Quantum Physics''.
Supported in part by the ERC Advanced Grant 227458
OACFT ``Operator Algebras and Conformal Field Theory''.
} \\
e-mail: {\tt bischoff@theorie.physik.uni-goettingen.de}\\
Institut f\"ur Theoretische Physik, Universit\"at G\"ottingen \\
Friedrich-Hund-Platz 1, 37077 G\"ottingen, Germany.\\
{\bf Yoh Tanimoto}%
\footnote{Supported
by Deutscher Akademischer Austauschdienst until August 2012,
by Hausdorff Institut f\"ur Mathematik until December 2012,
by Alexander von Humboldt Stiftung until March 2013,
and by Grant-in-Aid for JSPS fellows 25-205 since April 2013.} \\
e-mail: {\tt hoyt@ms.u-tokyo.ac.jp}\\
Graduate School of Mathematical Sciences, The University of Tokyo\\
and Institut f\"ur Theoretische Physik, Universit\"at G\"ottingen\\
JSPS SPD postdoctoral fellow
}
\begin{document}
\maketitle
\begin{abstract}
Our previous constructions of Borchers triples are extended to massless scattering with
nontrivial left and right components. A massless Borchers triple is constructed from
a set of left-left, right-right and left-right scattering functions.
We find a correspondence between massless left-right scattering S-matrices and
massive block diagonal S-matrices. We point out a simple class of S-matrices with
examples.

We study also the restriction of two-dimensional models to the lightray.
Several arguments for constructing strictly local two-dimensional nets are presented
and possible scenarios are discussed.
\end{abstract}

\tableofcontents

\section{Introduction}\label{introduction}
Here we further study our operator-algebraic approach to constructing quantum field models
in the two-dimensional spacetime. In the previous works we have established the general theory
of (wedge-local) massless excitations \cite{DT11, Tanimoto12-2} and constructed
several families of examples \cite{Tanimoto12-2, BT12}. It has been revealed that
from a pair of chiral components of conformal field theory and an appropriate
S-matrix one can construct the von Neumann algebra corresponding to the wedge-shaped region.
The operators in strictly local regions are to be determined through the intersection
of such wedges \cite{Borchers92}.
In our previous result, we considered only simple particle spectrum. Here we allow
multiple particle spectrum. Given a set of massless S-matrices, we construct
a Borchers triple, which is a weakened notion of Haag-Kastler net.
A corresponding massive result has been obtained in \cite{LS12}.
We show also that given a set of massless S-matrices, it is possible to construct
a massive Borchers triple. This provides a simple class of massive models.
In addition, with this transparent formulation we exhibit a family of concrete examples
of S-matrices, both massless and massive.
Finally, we consider a restriction of a two-dimensional model on the lightray.
A novel strategy to construct two-dimensional models is proposed and
several candidates for this program are discussed.

To integrable quantum field theory there is another approach,
the so-called form factor bootstrap program \cite{Smirnov92, FS94}. One takes a Lagrangian,
and after discussing its symmetry, one conjectures the S-matrix. The Hilbert space is
identified with the Fock space twisted by the S-matrix and the local operators are
obtained when one finds the set of matrix components which satisfy the so-called
form factor equations. This program has seen many interesting developments, including
form factors of several S-matrices (e.g.\! \cite{ZZ92, Bernard92}
for massless S-matrices and \cite{DMS95, MS97} for form factors).
In massless models there are so-called left-left, right-right and left-right S-matrices
\cite{Bernard92}. We formulate the properties of S-matrices in terms of operator algebras
and construct corresponding Borchers triples.
By using an analogous twist as \cite{Tanimoto13}, it turns out that the same set of S-matrices can be
used to construct a massive Borchers triple.
In our approach, we construct first one-dimensional Borchers triples (defined below)
using the left-left and right-right S-matrices and the two-dimensional Borchers triple
is obtained by twisting with the left-right S-matrix.
In addition, we find a simple class of S-matrices which contains
an infinite family of concrete examples.

Conversely, for a given two-dimensional model, one can
simply restrict it to the lightray. In this way, one obtains a one-dimensional Borchers
triple. The full two-dimensional theory
is remembered through a one-parameter semigroup of Longo-Witten endomorphisms.
Under this restriction, several conjectures have been made for integrable models,
for example, the $\mathrm{SU}(2)$-Thirring model should correspond to the $\mathrm{SU}(2)$-current
algebra \cite{ZZ92}, or an asymptotically free theory should correspond to the free current (c.f.\! \cite{BLM11, LS12}).
We are not going to prove these conjectures. Rather, we will argue that any of such correspondence
would lead to further new two-dimensional Haag-Kastler nets.
Although we still do not have any nontrivial example to which this program applies,
it could in principle go beyond integrable models in which the particle number is always conserved.

This paper is organized as follows. In Section \ref{preliminaries} we collect the notions
in the operator-algebraic approach to QFT, especially those oriented to scattering theory
and conformal field theory.
Section \ref{massless} treats massless integrable models. We define two-particle S-matrix
and construct the corresponding Borchers triples. It is shown that a class of massive S-matrices
can be used to construct massless S-matrices.
Then we observe that such a massless S-matrix can be turned into a massive S-matrix in Section \ref{massive}.
We exhibit the correspondence between one- and two-dimensional models in Section \ref{holography}.
Several existing conjectures are explained and a possible strategy for new two-dimensional models
is presented.
We gather open problems in Section \ref{outlook}.

Parts of this paper are based on the Ph.D.\! thesis of the author (M.B.) \cite{Bischoff12}.

\section{Preliminaries}\label{preliminaries}
Here we collect fundamental notions in the operator-algebraic approach to
scattering theory. Many of them are generalizations of the ones
which we considered before \cite{Tanimoto12-2, BT12}.
Some properties remain valid for such generalizations.

\subsection{Algebraic QFT and Borchers triples}\label{borchers-triples}
A {\bf Haag-Kastler net} $(\mathcal{A}, U, \Omega)$ is an axiomatization of local observables
in quantum field theory.
It is an assignment of a von Neumann algebra $\mathcal{A}(O)$ on a common Hilbert space $\mathcal{H}$
to each open region $O \subset \mathbb{R}^d$.
It should be covariant with respect to a unitary representation $U$
of the Poincar\'e group on $\mathcal{H}$ and possess an invariant ground state given by the vacuum vector $\Omega$.
The triple $(\mathcal{A},U,\Omega)$ is subject to standard axioms and considered as
a model of quantum field theory \cite{Haag96}. Each von Neumann algebra $\A(O)$ is considered
to be the algebra generated by the observables measured in the region $O$.
For example, if one has a quantum field in the sense of Wightmann given
by an operator valued distribution $\phi(f)$ acting on a Hilbert space $\H$, one obtains---provided the fields 
commute for functions with spacelike separated supports in a strong sense---a Haag-Kastler net on $\H$ by 
taking $\A(O) = \{\ee^{\ima \phi(f)} : \supp f\subset O\}''$.

It holds by the general Reeh-Schlieder argument that $\Omega$ is cyclic and separating 
for $\A(W_\mathrm{R})$, the algebra associated with the standard right-wedge $W_\R := \{(a_0,a_1) \in\RR^2: a_1 > |a_0|\}$.
Then there is a one-parameter group of unitaries $\{\Delta^{\ima t}\}$ canonically associated with the pair 
$(\A(W_\mathrm{R}),\Omega)$ by Tomita-Takesaki modular theory \cite{TakesakiII}.
These and the spacetime translations have the same commutation relation as that of Lorentz boosts and translations,
and in many cases they actually coincide, $\Delta^{\ima t} =U(\Lambda(-2\pi t))$
(Bisognano-Wichmann property).

It seems quite difficult to construct such an infinite family $\{{\mathcal A}(O)\}$ of von Neumann algebras
with compatibility conditions. Actually, Borchers proved that for $d=2$,
it is enough to consider a single von Neumann algebra $\M$
which is associated with $W_\R$, 
the spacetime translations $T$ and an invariant vector $\Omega$.
Such a triple $(M, T, \Omega)$ subject to several requirements
is called a Borchers triple and we will give its formal definition below.

If $({\mathcal A},U,\Omega)$ is a Haag-Kastler net, then $({\mathcal A}(W_{\mathrm R}), U|_{{\mathbb R}^2}, \Omega)$ is a Borchers triple,
and we consider the restriction of $U$ to the spacetime translations.
Conversely, if one has a Borchers triple $({\mathcal M},T,\Omega)$, it is possible to define a net
as follows:
one first defines a net for every wedge by 
\begin{align}
	\label{eq:WedgeNet}
	\A(W_{\mathrm R} +a) &= \Ad T(a)(\M)\,,&
	\A(W_{\mathrm L} +b) &= \Ad T(b)(\M')\,,
\end{align}
 where $W_{\mathrm L}$ is the standard left-wedge.
To pass to bounded regions one just has to take intersections, more precisely
any double cone (diamond) in two dimensional spacetime can be represented as the intersection
of two-wedges: $(W_{\mathrm R}+a)\cap (W_{\mathrm L}+b) =: D_{a,b}$, where $W_{\mathrm L}$ is the standard left-wedge.
Then the von Neumann algebras for double cones $D_{a,b}$ are defined 
by ${\mathcal A}(D_{a,b}) := {{\rm Ad\,}} T(a)({\mathcal M}) \cap {{\rm Ad\,}} T(b)({\mathcal M}')$.
For a general region $O$ one takes the union from the inside:
${\mathcal A}(O) := \left(\bigcup_{D_{a,b}\subset O} {\mathcal A}(D_{a,b})\right)''$.
Furthermore, one can extend the representation of the translation group to a representation of the whole Poincar\'e group by using the Tomita-Takesaki theory of von Neumann algebras. 
Namely one defines the representation guided by the Bisagnono-Wichmann property above and Borchers' theorem ensures that this really defines a representation of the 
Poincar\'e group \cite{Borchers92}.
More precisely, the one-parameter unitary group $\{\Delta^{\ima t}\}$ canonically associated with the pair
of a von Neumann algebra $\M$ and $\Omega$ represents the Lorentz boosts.

Then one can show that this ``net'' $({\mathcal A}, U, \Omega)$
satisfies almost all of the properties of Haag-Kastler net.
But, while the wedge algebras are by definition always sufficiently large, i.e.\ they generate the whole Hilbert space $\H$ from the vacuum $\Omega$,
it is in general difficult to show that for local algebras ${\mathcal A}(D_{a,b})$
and it can actually fail \cite[Theorem 4.16]{Tanimoto12-2}.
But if it is the case, then the triple indeed defines a Haag-Kastler net by the above structure. 
This program has been accomplished in some cases and obtained families of interacting models
\cite{Lechner08, Tanimoto13}.
It might happen that $\A(D_{a,b})$ just contains the scalars and one would not have any local observables.
If there are non-trivial local observables in $\A(D_{a,b})$ one gets at least a Haag-Kastler net on a smaller Hilbert space $\H_0=\overline{\A(D_{a,b})\Omega}$.

A {\bf Borchers triple} on a Hilbert space $\H$ is a triple $(\M, T, \Omega)$
of a von Neumann algebra $\M$, a unitary representation $T$ of $\RR^2$
and a unit vector $\Omega$, such that
\begin{enumerate}[ {(}1{)} ]
 \item If $a\in W_R$, then $\Ad T(a)(\M) \subset \M$.
 \item The joint spectrum of $T$ is contained in the closed forward lightcone
 $\overline{V_+} := \{(a_0,a_1)\in\RR^2: a_0 \ge |a_1|\}$.
 \item $\Omega$ is cyclic and separating for $\M$.
\end{enumerate}
In the sense explained above, a Borchers triple gives a Poincar\'e covariant, wedge-local net defined by equation (\ref{eq:WedgeNet}) and can be considered to be a ``net of observables localized in wedges''.
If $\Omega$ is cyclic for the von Neumann algebra $\M \cap \Ad T(a)(\M')$ for any
$a \in W_\R$, one can construct a Haag-Kastler net on the original Hilbert space $\H$ and in this case we say
that the Borchers triple $(\M,T,\Omega)$ is {\bf strictly local}.
In Sections \ref{massless} and \ref{massive} we construct Borchers triples and
Section \ref{holography} is concerned with strictly local triples.

\subsubsection*{The massive scalar free field}
The simplest Borchers triple is constructed from the simplest quantum field.
The one-particle Hilbert space of the free scalar field of mass $m>0$
is given by $\H_m := L^2(\RR, \dd\theta)$ and 
the translation acts by $(T_m(a)\psi)(\theta) = \ee^{\ima p_m(\theta)\cdot a}\psi(\theta)$,
where $p_m(\theta) := (m\cosh(\theta), m\sinh(\theta))$ parametrizes the mass shell.
We need the unsymmetrized Hilbert space $\H^\Sigma_m := \bigoplus \H_m^{\otimes n}$ and
the symmetrized Hilbert space $\Hr := \bigoplus P_{n,\mathrm{sym}} \H_m^{\otimes n}$,
where $P_{n,\mathrm{sym}}$ is the projection onto the symmetric subspace.

Let $a_{\mathrm{r}}^\dagger$ and $a_{\mathrm{r}}$ be the creation and annihilation operators as usual
(see \cite[Section 2.3]{Tanimoto13}. In our notation,
$a_{\mathrm{r}}^\dagger(\psi)$ is linear and $a_{\mathrm{r}}(\psi)$ is antilinear with respect to $\psi$).
The (real) free field $\phi_{\mathrm{r}}$ is defined by
\[
\phi_{\mathrm{r}}(f) := a_{\mathrm{r}}^\dagger(f^+) + a_{\mathrm{r}}(J_mf^-), \;\;\;\;\;\;\;\;\; f^\pm(\theta) = \frac{1}{2\pi}\int \dd^2af(a)\ee^{\pm \ima p_m(\theta)\cdot a},
\]
where $f$ is a test function in $\mathscr{S}(\RR^2)$ and $J_m\psi(\theta) = \overline{\psi(\theta)}$.
Our von Neumann algebra is
\[
\Mr := \{\ee^{\ima\phi_{\mathrm{r}}(f)}: \supp f \subset W_\R\}''.
\]
The translation on the full space is the second quantized representation
$\Trr := \Gamma(T_m)$ and there is the Fock vacuum vector $\Omr \in \Hr$.
This triple $(\Mr,\Trr,\Omr)$ is the Borchers triple of the free field.
Of course this is strictly local and the corresponding net is the familiar
free field net. A more abstract definition of this free field construction
starting from a general positive energy representation of the Poincar\'e is given in \cite{BGL02}.

\subsubsection*{Examples from integrable models}
The form factor bootstrap program, an approach to integrable quantum field theory,
can be briefly summarized as follows \cite{Smirnov92}. First a model with infinitely
many conserved current is considered.
The scattering matrix turns out to be factorizing, then the explicit form of it is
speculated by a symmetry argument. Finally, one finds solutions of the so-called
form factor equation, which is given in terms of the two-particle scattering function.
A solution of the form factor equation is a series of functions. It is supposed
to serve as the matrix coefficients of a local observable. The convergence of the
series as an operator is expected in a wide class of models but remains open.

An alternative approach has been initiated by Schroer \cite{Schroer97, Schroer99}
and worked out by Lechner \cite{Lechner08}. In this approach, given an S-matrix,
the operators localized in a wedge are constructed and the local observables
are obtained as the intersection of left and right wedges. The determination of
the intersection, which in the form factor program would correspond to finding
form factors (and proving the convergence), has been done with the help of
operator algebraic methods including the Tomita-Takesaki theory of von Neumann
algebras \cite{BL04}.

For the case of single species of particle of mass $m>0$ (the scalar case) treated in \cite{Lechner08}
one takes a bounded analytic function
$S_2(\theta)$ on the strip $\RR+\ima(0,\pi)$, continuous on the boundary,
such that
\[
S_2(\theta)^{-1} = \overline{S_2(\theta)} = S_2(-\theta) = S_2(\theta+i\pi)
\]
for $\theta \in \RR$.

The one-particle space $\H_1$ is the same as that of the free field.
On $n$-particle space one defines the $S_2$-permutation by
\[
(D_{S_2,n}(\tau_j)\Psi)(\theta_1,\cdots, \theta_n) = S_2(\theta_{j+1}-\theta_j)\Psi(\theta_1,\cdots,\theta_{j+1},\theta_j,\cdots,\theta_n).
\]

This time $P_{n,S_2}$ is the orthogonal projection onto the subspace of $\H_1^{\otimes n}$ invariant
under $\{D_{S_2,n}(\tau_j): i \le j \le n\}$.
We take the Hilbert space $\H_{S_2} := \bigoplus P_{n,S_2} \H_1^{\otimes n}$,
the representation $T_{S_2}$ is the second quantized promotion of $T^1$
and the Fock vacuum is denoted by $\Omega_{S_2}$.
The creation and annihilation operators are given by
$(z_{S_2}^\dagger(\psi)\Phi)_n = \sqrt n P_{n,S_2}(\psi\otimes \Phi_{n-1})$
and $z_{S_2}(\psi) = z_{S_2}^\dagger(\psi)^*$. 
For a test function $f$ on $\RR^2$, the wedge-local field is defined also as
\[
\phi_{S_2}(f) := z_{S_2}^\dagger(f^+) + z_{S_2}(J_1f^-), \qquad 
f^\pm(\theta) = \frac{1}{2\pi}\int \dd^2af(a)\ee^{\pm \ima p_m(\theta)\cdot a}\,\text.
\]

The von Neumann algebra $\M_{S_2}$ is given by
\[
\M_{S_2} := \{\ee^{\ima\phi_{S_2}(f)}: \supp f \subset W_\L\}'.
\]
The triple $(\M_{S_2}, U_{S_2}, \Omega_{S_2})$ is a Borchers triple \cite{Lechner03}
and strictly local if $S_2$ is regular and fermionic ($S_2(0)=-1$) \cite{Lechner08}.

\subsection{One-dimensional Borchers triple}\label{half-line}
Let $\H_0$ be a Hilbert space. A triple $(\M_0, T_0, \Q_0)$ of a von Neumann
algebra $\M_0$, a unitary representation $T_0$ of $\RR$ with positive generator
and a unit vector $\Q_0$ is said to be a {\bf one-dimensional Borchers triple} if
$\Q_0$ is cyclic and separating for $\M_0$ and it holds that $\Ad T_0(t)(\M_0) \subset \M_0$
for $t\ge 0$. Note that this notion is equivalent to that of {\bf half-sided modular inclusion}
\cite{Wiesbrock93, AZ05} if one considers the inclusion $\Ad T_0(1)(\M_0) \subset \M_0$.

If $\Q_0$ is cyclic for the intersection $\M_0 \cap \Ad T_0(1)(\M_0)$, then we say
that the triple $(\M_0, T_0, \Q_0)$ is {\bf strictly local}. The corresponding
notion in half-sided modular inclusion is the {\bf standardness}.
If one has a strictly local one-dimensional Borchers triple, then one can construct
a M\"obius covariant net of von Neumann algebras on $S^1$ (see below), in which $\M_0$ and $T_0$
correspond to the algebra of the half-line $\RR_+$ and the translation, respectively
\cite{GLW98}.

After this remark it is natural to introduce the following concept (see \cite{LW11, Tanimoto12-2}):
a {\bf Longo-Witten endomorphism} of the triple $(\M_0, T_0, \Q_0)$ is an endomorphism
of $\M_0$ which is implemented by a unitary $V_0$, which commutes with $T_0$ and
preserves the vacuum state $\<\Q_0, \cdot\,\Q_0\>$. If we require that $V_0\Q_0 = \Q_0$,
such an implementation is unique.

\subsubsection*{Examples from nets}
An important class of examples comes from M\"obius covariant nets on $S^1$.
A {\bf M\"obius covariant net} of von Neumann algebras on $S^1$ defined on $\H_0$
is a triple $(\A_0, U_0, \Q_0)$, where $\A_0$ assigns a von Neumann algebra $\A_0(I)$
to each proper open interval $I\subset S^1$,
$U_0$ is a unitary representation of the M\"obius group $\psl2r$ and $\Q_0$,
which satisfy certain properties (see the preliminary sections in \cite{Tanimoto12-2, BT12}).
Then $(\A_0(\RR_+), U_0|_{\RR}, \Q_0)$ is a one-dimensional, strictly local
Borchers triple, where $\RR_+ \subset \RR$ is understood as a subset of $S^1$ by the
stereographic projection and $U_0$ is restricted to the translation subgroup of $\psl2r$
under this identification. Conversely, if one has a strictly local triple, one can
construct a M\"obius covariant net. 
The correspondence is one-to-one if one
assumes the M\"obius covariant nets to be strongly additive \cite{GLW98}.

Similarly, if we take a (two-dimensional) Borchers triple $(\M, T, \Q)$,
then one can consider the restriction of $T$ to the positive lightray $\{(t,t)\in\RR^2: t\in \RR\}$,
which we denote by $T_+$.
It is immediate that the triple $(\M,T_+,\Q)$ is a one-dimensional Borchers triple
($W_\R$ is by definition an open wedge, hence does not include the lightrays,
but the inclusion relation for Borchers triple is immediate from the strong continuity of $T$
and strong closedness of $\M$).
We will discuss this class with examples in detail in Section \ref{holography}.

\subsection{Massless scattering theory}\label{scattering}
Usually the existence of massless particles is a source of difficulty in
scattering theory. We have seen that an additional assumption, 
asymptotic completeness, greatly reduces the problem \cite{DT11, Tanimoto12-1, Tanimoto12-2}.
Of particular importance is the result \cite[Section 3]{Tanimoto12-2} that
a Haag-Kastler net which is asymptotically complete with respect to waves
(the corresponding notion of massless particles in the two-dimensional spacetime)
can be easily reconstructed from its asymptotic (free) behavior and the S-matrix.
In this paper we are concerned only with such models.

\subsubsection*{Borchers triples by tensor product}
A (two-dimensional) Borchers triple can be constructed out of a pair of one-dimensional
Borchers triples $(\M_\pm, T_\pm, \Q_\pm)$ as follows.
Let $(t_+,t_-)$ be the lightray coordinates of $\RR^2$, where $t_+ = t_0 - t_1$
and $t_- = t_0 + t_1$ (the indices might look unnatural, but are consistent with
the scattering theory \cite{Buchholz75, DT11}).
One takes a triple $(\M, T, \Omega)$ where
\begin{itemize}
 \item $\M := \M_+'\otimes \M_-$,
 \item $T(t_+,t_-) = T_+(t_+)\otimes T_-(t_-)$,
 \item $\Q = \Q_+\otimes\Q_-$.
\end{itemize}
Then it is immediate to see that this is a Borchers triple.
The representation $T$ is said to contain waves, in the sense that
there are nontrivial spectral projections concentrated in the lightrays.
This triple naturally turns out not to interact, namely the S-matrix is
the identity operator $I$ \cite{Tanimoto12-2}.

\begin{figure}[ht]
    \centering
    \definecolor{mygrey}{HTML}{d2d2d2}
    \begin{tikzpicture}[scale=0.8]
        \begin{scope}
            \fill[color=mygrey,decoration={random steps, segment length=6pt, amplitude=2pt}]
           (0,0)--(3.8,3.8) decorate{-- (3.8,-3.8)} --(0,0);
            \draw [dashed] (-3.5,-3.5)--(3.5,3.5);
            \draw [dashed] (3.5,-3.5)--(-3.5,3.5) node [right] {$t_+$};
            \draw [thick, ->] (-4,0)--(4,0) node [above left] {$t_1$};
            \draw [thick, ->] (0,-4)--(0,4) node [below right] {$t_0$};
            \node at (2.25,.725) {$W_\mathrm{R}$};
            \draw [ultra thick] (0,0)-- node [right] {$\scriptstyle \1 \otimes \M_-$} (3.5,3.5) node [left] {$t_-$};
            \draw [ultra thick,dotted] (3.5,3.5)--(4,4);
            \draw [ultra thick] (0,0)-- node [ right]  {$\scriptstyle \M_+' \otimes \1$} (3.5,-3.5);
            \draw [ultra thick,dotted] (3.5,-3.5)--(4,-4);
        \end{scope}
        \begin{scope}[shift={(9,0)}]
            \fill[color=mygrey,decoration={random steps, segment length=6pt, amplitude=2pt}]
            (1.5,1)--(3.8,3.3) decorate{-- (3.8,-1.3)} --(1.5,1);
   	    \draw [dashed] (1.15,1.35)--(1.5,1)--(0.15,-.35);
            \draw [dashed] (4,3.5)--(1.5,1)--(4,-1.5);
            \node at (2.25,1.25) [right] {$W_\mathrm{R} +(t_+,t_-)$};
            \draw [ultra thick] (1.25,1.25)-- node [right] {$\scriptstyle (\1\otimes \mathrm{Ad\,} T_-(t_-)(\M_-)$} (3.5,3.5) node [left] {$t_-$};
            \draw [ultra thick,dotted] (3.5,3.5)--(4,4);
            \draw [ultra thick] (0.25,-0.25)--node [right] {$\scriptstyle \mathrm{Ad\,} T_+(t_+)(\M_+')\otimes  \1$} (3.5,-3.5);
            \draw [ultra thick,dotted] (3.5,-3.5)--(4,-4);
            \draw [thick, ->] (-4,0)--(4,0) node [above left] {$t_1$};
            \draw [thick, ->] (0,-4)--(0,4) node [below right] {$t_0$};
           \draw [dashed] (-3.75,-3.75)--(4,4);
            \draw [dashed] (4,-4)--(-3.75,3.75) node [right] {$t_+$};
        \end{scope}
    \end{tikzpicture}
    \caption{On the definition of the tensor product Borchers triple}
    \label{fig:wedge-local net}
\end{figure}
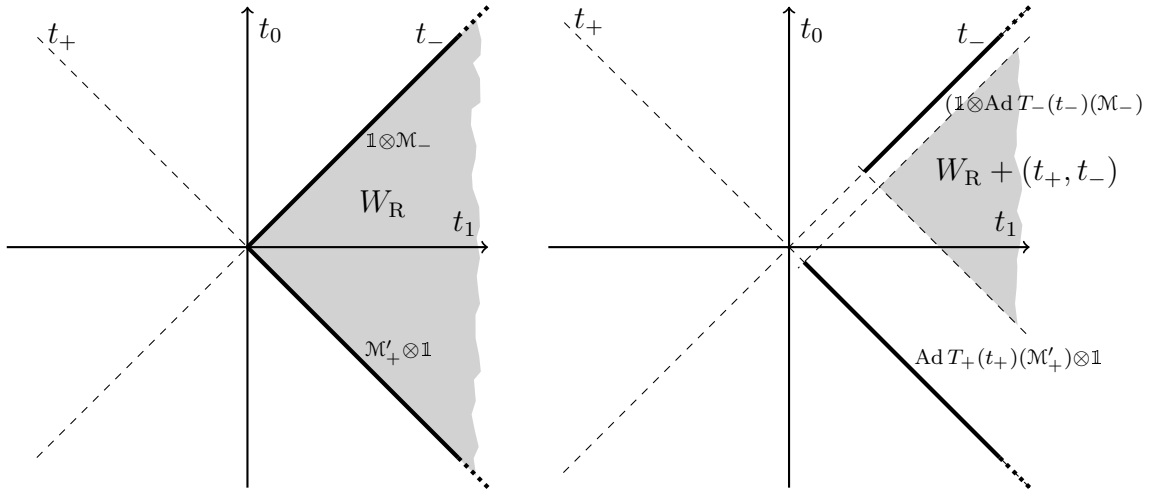

\subsubsection*{How to construct interacting models}
We do not repeat the definition of asymptotic completeness for waves \cite{Buchholz75, DT11}.
By repeating the proofs of \cite[Section 3]{Tanimoto12-2}\cite[Proposition 2.2]{BT12},
one can show the following.

\begin{proposition}\label{pr:massless-structure}
There is a one-to-one correspondence between
\begin{itemize}
 \item asymptotically complete (for massless waves) Borchers triples $\{(\M,T,\Q)\}$,
 \item 7-tuples $\{((\M_+, T_+, \Q_+), (\M_-, T_-, \Q_-), S)\}$, where $(\M_\pm, T_\pm, \Q_\pm)$ are
 one-dimensional Borchers triples and $S$ is a unitary operator on $\H_+\otimes\H_-$ commuting with
$T_+\otimes T_-$, leaving $\H_+\otimes \Omega_-$ and $\Omega_+\otimes\H_-$
pointwise invariant,
such that $x'\otimes \1$ commutes with $\Ad S(x\otimes \1)$ where
$x \in \M_+$ and $x' \in \M_+'$, and $\Ad S(\1\otimes y)$ commutes with $\1\otimes y'$ where
$y \in \M_-$ and $y \in \M_-'$.
\end{itemize}
\end{proposition}
The correspondence is given by
\begin{itemize}
\item $\M := \M_+'\otimes \1 \vee \Ad S(\1\otimes \M_-)$,
\item $T(t_+,t_-) := T_+(t_+)\otimes T_-(t_-)$,
\item $\Q := \Q_+\otimes \Q_-$.
\end{itemize}
Indeed, the properties of net (strict locality) are used only to show
the M\"obius covariance of the one-dimensional components, which we do not
claim here and the rest of the proofs works.

Our program to construct massless Borchers triples is now split into two steps:
first prepare a pair of one-dimensional Borchers triples, then
find an appropriate operator $S$ to make them interact.
We carry out this program in Section \ref{massless}. We do not investigate
strict locality in the present paper.

\section{Massless models with nontrivial scattering}\label{massless}
Here we construct massless Borchers triples following the program described
in Section \ref{scattering}. As an input we take so-called left-left, right-right
and left-right scattering matrices (c.f.\! \cite{Bernard92}).

Usually the form factor bootstrap program is carried out for massive models.
Massless limit makes worse the behavior of the form factors in the momentum
space and even the fundamental ``local commutativity theorem'' \cite{Smirnov92}
fails when applied to concrete cases.
As for the operator algebraic approach, the modular nuclearity has
been proved through a careful estimate \cite{Lechner08},
which will no longer be valid for the massless case.

Yet in operator-algebraic approach, half of the program can be carried out: one can
construct certain operators to be interpreted as observables in a wedge.
This has been done in \cite{LS12} for the massive case with multiple particle spectrum
and in \cite{Tanimoto12-2, BT12} for the massless case with simple spectrum.
In this Section we exhibit a massless construction which includes several kinds of particles.
It is also interesting to observe at which point the Yang-Baxter equation enters.

\subsection{Scattering matrices and operators}
As in massless bootstrap program, we need two kinds of input: left-left and right-right
scattering and left-right scattering.
While the former governs the asymptotic behavior of the model,
the latter is directly related to the S-matrix.

\subsubsection{Scattering matrices for chiral parts}
One-dimensional Borchers triples can be obtained 
by second quantization of so-called standard pairs,
similarly to the algebraic construction of massive models 
with factorizing $S$-matrices \cite{LS12} and the free field construction in \cite{BGL02}. 
This will be done on a $R$-symmetric Fock space (defined in Section \ref{sec:2ndQBP}),
where $R$ is a certain operator.
We give an abstract definition for suitable operators $R$ and characterize them in terms of 
usual scattering matrices. They are called left-left or right-right scattering operator
in physics literature from a formal similarity to S-matrix,
but the physical meaning of $R$ remains unclear, c.f.\! \cite{BLM11}.

Let $\H$ be a Hilbert space.
For operator $A \in \B(\H\otimes\H)$ we denote by 
$A_{ij}$ the operator on $\B(\H^{\otimes n})$ ($n\geq i,j$) which 
acts by $A$ on the product of the $i$-th and the $j$-th tensor factors.
For example, if $A = A_1\otimes A_2$, then
$A_{ij} = \1\otimes\cdots\otimes \underset{i\mbox{-th}}{A_1}\otimes\cdots\otimes\underset{j\mbox{-th}}{A_2}\otimes\cdots\1$.

A closed, real linear subspace $H\subset\H$ with 
$H\cap\ima H=\{0\}$ and $\overline{H+\ima H}=\H$ is called {\bf standard}.
We denote by 
$H'=\{x\in\H: \Im \langle H,x\rangle=0\}$, where $\Im$ is the imaginary
part, the symplectic complement of a closed real 
linear space,
which is standard if and only if $H$ is standard. 
With a standard subspace $H$ we can associate modular objects, 
i.e.~an antiunitary involution $J_{H}$ and 
a unitary one-parameter group $\{\Delta_{H}^{\ima t}\}_{t\in\RR}$ by the 
polar decomposition $S_{H}=J_{H}\Delta_{H}^{1/2}$ of the 
densely defined, closed, antilinear involution $S_{H}:f+\ima g\mapsto f-\ima g$ 
for $f,g\in H$. A (simpler) one-particle version of Tomita-Takesaki
theory says that $\Delta_{H}^{\ima t} H=H$ and $J_{H}H=H'$ hold \cite{Longo08}.

Let $H$ be a standard subspace of a Hilbert space 
$\H$ and let us assume that there exists a one-parameter group 
$T(t)=\ee^{\ima t P}$ on $\H$ such that
\begin{itemize}
    \item $T(t) H \subset H$ for all $t\geq 0$,
    \item $P$ is positive and $P$ has no point spectrum in $0$.
\end{itemize}
Then we call the pair $(H,T)$ a ({non-degenerate})
\textbf{standard pair}. 
A standard pair is called 
\textbf{irreducible} if it cannot be written as a non-trivial direct sum 
of two standard pairs.

There exists a unique (up to unitary equivalence) irreducible standard pair
$(H_0,T_0)$ whose ``Schr\"odinger representation'' is given as follows.
We realize  $(H_0,T_0)$ on $\H_0=L^2(\RR)$ 
and $T_0(t)=\ee^{\ima t P_0}$, where $Q_0=\ln P_0$ with 
$(\ee^{\ima t Q_0}f)(q) =  \ee^{\ima t q}f(q)$.
A function $f \in L^2(\RR)$ is in $H_0$ if and only if $f$ admits an analytic continuation
on the strip $\RR + \ima(0,\pi)$, such that
for every $a\in(0,\pi)$ it is: $f(\slot + \ima a)\in L^2(\RR)$ with 
boundary value $f(q+\ima \pi)=(J_{H_0}f)(q):=\overline{f(q)}$.
One defines $(\Delta_{H_0}^{-\ima s } f)(q)=f(q+2\pi s)$ and it can be easily checked 
that $(J_{H_0},\Delta_{H_0}^{\ima t})$ are indeed the modular objects for $H_0$
\cite{Longo08}.

For a standard pair $(H,T)$ we give an abstract definition of an operator $R$,
which encodes the two-particle scattering process. 
\begin{definition} 
    \label{def:TwoParticleScatteringOperator}
    Let $(H,T)$ be a standard pair in $\H$.
	Let $\S(H,T)$ be the set of all 
    unitary operators $R\in \B(\H\otimes\H)$
    such that the following properties hold:
    \begin{enumerate}[ {(}1{)} ]
        \item Reflection property: $R_{21}=R^\ast$.
        \item Yang-Baxter equation: $R_{12}R_{13}R_{23} = R_{23}R_{13}R_{12}$ on $\H^{\otimes 3}$.
        \item Translation covariance: $[R,T(t_1)\otimes T(t_2)] =0$.
        \item Dilation+TCP covariance:
            $[\Delta_H^{\ima t}\otimes \Delta_H^{\ima t}, R] = 0$
		and $R(J_H \otimes J_H) = (J_H\otimes J_H)R^\ast$.
    	\item Half-line locality:
  		$\langle g'\otimes \eta, R(f\otimes \xi)\rangle = \langle f\otimes \eta, R^\ast (g'\otimes \xi)
	\rangle$ for 
        all $f\in H, g\in H'$ and $\xi,\eta\in\H$ or equivalently:
            the operator $A^R_{f,g'}$ 
            defined by 
            \begin{align*}
                \xi\mapsto \sum_k \langle g' \otimes e_k , R(f\otimes \xi)\rangle \cdot e_k
            \end{align*} with $\{e_k\}$ an orthonormal basis of $\H$,
            is self-adjoint for all $f\in H, g'\in H'$.
    \end{enumerate}
\end{definition}
We will see that the locality assumption follows from the requirement that, 
on two-particle level, certain generators of the wedge-algebra fulfill half-line
locality in Lemma \ref{lm:FieldLocality}. 

We remember that each (non-degenerate) standard pair $(H,T)$ is a direct sum of the 
unique irreducible standard pair $(H_0,T_0)$ \cite{LW11}. A standard pair 
with multiplicity $n$ can be given as follows.
We can choose a Hilbert space $\K$ with $\dim\K=n$ and 
$\H=\H_0\otimes \K\cong L^2(\RR,\CC^n)$  and $T(t)=\ee^{\ima t P}:=T_0(t)\otimes \1$,
$\Delta_H^{\ima t}=\Delta_{H_0}^{\ima t}\otimes \1$.
To make contact with the physics literature, we choose some orthonormal basis
indexed by $\{\a\}$
of $\K$ and an involution $\alpha\mapsto \bar \alpha$ on the index set and define the antiunitary involution $J_H$ to be
\begin{align}
	(J_Hf)^\alpha(q)= \overline{ f^{\bar\alpha}(q)}
	~.
\end{align}
Then a function $f=(f^\alpha) \in L^2(\RR,\CC^n)$ is in $H$ if and only if $f$ admits an analytic continuation
on the strip $\RR + \ima(0,\pi)$, such that
for every $a\in(0,\pi)$ it is: $f^\alpha(\slot + \ima a)\in L^2(\RR)$ with 
boundary value $f_\alpha(q+\ima \pi)=\overline{f_{\bar\alpha}(q)}$. Every standard pair 
with finite multiplicity is of this form.

Due to unitarity, translation covariance and the fact that $R$ commutes with $\Delta_H^{\ima t}\otimes 
\Delta_H^{\ima t}$, a two-particle scattering operator is given by the spectral calculus by $\und{R}(Q_1-Q_2)$,
where $Q_1=Q\otimes \1$, $Q_2=\1\otimes Q$, $Q = \ln P$, $P$ is the generator of $T$ and
$q\mapsto \und{R}(q)$ is a operator-valued function from $\RR$ to $\B(\K\otimes\K)$ which is 
unitary almost everywhere. By fixing a basis on $\K$, we can represent $\und{R}(q)$ as a matrix
$\und{R}^{\a\b}_{\g\d}(q)$ (almost everywhere). In the above representation this reads
\begin{align}
	\label{eq:ansatz1}
	(R\xi)^{\alpha\beta}(q_1,q_2) 
	=\und{R}^{\alpha\beta}_{\gamma\delta}(q_1-q_2) \xi^{\gamma\delta}(q_1,q_2)
	=:\und{S}^{\beta\alpha}_{\gamma\delta}(q_1-q_2) \xi^{\gamma\delta}(q_1,q_2)
\end{align}
where it is sometimes common to use the matrix valued function $q\mapsto \und{S}(q)$ with interchanged indices, c.f.~\cite{LS12}.

{\it Note}. In the following, symbols with underline denote matrix-valued functions
or equivalently functions with operator-value on a finite dimensional Hilbert space.

Let us define the operator-valued function $\RR\ni q\mapsto R^{\delta}_{\beta}(q)\in\B(\H)$
by 
\begin{align}
	\label{eq:R_q}
	 R^{\b}_{\d}(s)&:=\und{R}^{\slot\b}_{\slot\d}(Q-s), \text{~i.e.}\quad
	(R^\b_\d(s)\xi)^\a(q) = \sum_{\g} \und{R}^{\a\b}_{\g\d}(q-s)\xi^\g(q)
	\,\text,
\end{align}
where $\xi \in \H$.
The partial disintegration of $R$ reads
\[
R = 
  \sum_{\beta, \delta}\int R^\beta_\delta(q)\otimes \dd E(q)^\beta_\delta
	\,\text,
\]
where $\dd E^\b_\d = \dd E_0 \otimes E^\b_\d$ and $\dd E_0$ is the spectral measure
of $Q_0=\ln P_0$ and $E^\b_\d$ is the operator corresponding on the fixed basis $\{\xi_\a\}$
to the matrix which has the value $1$ in $(\b,\d)$-component and $0$ in the others.

Before giving a characterization of the operators $R\in\S(H,T)$ 
we prove the following Lemma, which will reduce the argument of half-line locality 
to two-particle processes.
\begin{lemma}
    \label{lem:WedgeLocality}
    Let $(H,T)$ be a standard pair, $R\in\S(H,T)$ 
	and 
    $\tilde R=R_{1,n+1}R_{1,n}\cdots R_{1,2}$ on $\H^{\otimes n+1}$. Then the operator
    $A^{\tilde R}_{f,g'}\in\B(\H^{\otimes n})$ given by 
    \begin{align*}
        \xi\mapsto \sum_{\tilde k} \< g'\otimes e_{\tilde k}, \tilde R(f\otimes \xi)\> \cdot e_{\tilde k}
    \end{align*}
    is self-adjoint for all $f\in H, g\in H'$,
    where $\{e_{\tilde k}\}$ is a basis on $\H^{\otimes n}$.
\end{lemma}
\begin{proof}
Because every standard pair is just a direct sum of the irreducible
standard pairs, we may assume $\bar\alpha=\alpha$ in the above decomposition. 
We can write $R$ as 
\begin{align*}
        R&=\sum_{\beta,\delta}\int R^\beta_{\delta}(q)\otimes \dd E(q)^\beta_\delta~, \\ 
        R^\ast&=\sum_{\beta,\delta}\int (R^\b_{\d}(q))^\ast\otimes \dd E(q)^\d_\b 
        =\sum_{\beta,\delta}\int (R^\delta_{\beta}(q))^\ast\otimes \dd E(q)^\beta_\delta
	\, \text.
\end{align*}
Then by the assumption that $R\in \S(H,T)$, for all $f\in H, g'\in H'$ we have 
\begin{align*}
        \sum_{\beta,\delta}\int \<g', R^\beta_{\delta}(q)f\> \dd E(q)^\beta_\delta
        &=
        A^R_{f,g'} 
        =
        (A^R_{f,g'})^\ast  =\sum_{\beta,\delta}\int \<R^\d_{\b}(q) f,g' \>
        \dd E(q)^\beta_\delta
	\,\text,
\end{align*}
which is equivalent to $R^\beta_{\delta}(q)S_H\subset S_H R^\delta_{\beta}(q)$ 
for almost all $q$ by Lemma \ref{lem:LW}.
But this implies that also
$R^{\beta_1}_{\delta_1}(q_1)\cdots R^{\beta_n}_{\delta_n}(q_n)S_H\subset
S_H R^{\delta_1}_{\beta_1}(q_1)\cdots R^{\delta_n}_{\beta_n}(q_n)$
holds, hence using again Lemma \ref{lem:LW} the equality of the following two operators follows
    \begin{align*}
        A^{\tilde R}_{f,g'} &= \sum_{\b_1,\ldots,\b_n,\d_1,\ldots,\d_n}\int 
        \<g', R^{\beta_1}_{\delta_1}(q_1)\cdots R^{\beta_n}_{\delta_n}(q_n) f\>
        \dd  E(q_n)^{\beta_n}_{\delta_n}\otimes \cdots \otimes \dd E(q_1)^{\beta_1}_{\delta_1}
        &\text{and}\\
        (A^{\tilde R}_{f,g'})^\ast &= \sum_{\b_1,\ldots,\b_n,\d_1,\ldots,\d_n}\int 
        \overline{\<g', R^{\d_1}_{\b_1}(q_1)\cdots R^{\d_n}_{\b_n}(q_n) f\>}
        \dd  E(q_n)^{\beta_n}_{\delta_n}\otimes \cdots \otimes \dd  E(q_1)^{\beta_1}_{\delta_1}
        \,\text,
    \end{align*}
 which proves the claim.
\end{proof}

We characterize the two-particle scattering operators $R$ in terms of matrix-valued function
and show that they indeed come from two-particle scattering matrices (c.f.\! \cite{LS12}). 

\begin{proposition} 
	\label{prop:Rmatrix}
	Let $(H,T)$ be a standard pair with finite multiplicity.
	Then $R\in\S(H,T)$ if and only if $R$ comes from a matrix valued functions $(\und{S}^{\alpha\beta}_{\gamma\delta})$ as in  (\ref{eq:ansatz1}),
        fulfilling the following relations:
\begin{enumerate}[ {(}1{)} ]
	\item \emph{Unitarity:} $\und{S}(q)$ is an unitary matrix
	for almost all $q\in\RR$.
	\item \emph{Hermitian analyticity:} $\und{S}(-q)=\und{S}(q)^\ast$	 
	for almost all $q\in\RR$.
	\item \emph{Yang-Baxter equation:}
	\begin{align*}
		\und{S}(q)_{12}\und{S}(q+q')_{23}\und{S}(q')_{12}&=
		\und{S}(q')_{23}\und{S}(q+q')_{12}\und{S}(q)_{23}\,\text.
	\end{align*}
	\item \emph{TCP:} $\und{S}^{\alpha\beta}_{\gamma\delta}(q)=\und{S}^{\bar\delta\bar\gamma}_
		{\bar\beta\bar\alpha}(q)$ for almost all $q\in\RR$.
	\item \emph{Analyticity:} $q\mapsto \und{S}(q)$ is boundary value of a bounded analytic function
	on $\RR+\ima(0,\pi)$.
	\item \emph{Crossing symmetry:} 
	$\und{S}^{\alpha\beta}_{\gamma\delta}(\ima \pi-q)= \und{S}^{\bar\gamma\alpha}_{\delta\bar\beta}(q)$.
\end{enumerate}
\end{proposition}
\begin{proof}
As discussed above the ansatz in equation (\ref{eq:ansatz1}) is equivalent to unitarity,
translation covariance and the fact that $R$ commutes with $\Delta_H^{\ima t}\otimes 
\Delta_H^{\ima t}$. It is straightforward to check that hermitian analyticity of $\und{S}(\slot)$ is equivalent 
to the reflection property of $R$; the property $R(J_H\otimes J_H)=(J_H\otimes J_H) R^\ast$ is equivalent to TCP, and 
Yang--Baxter equation of $R$ with the one for the matrices $\und{S}(q)$.

Using $R^{\delta}_\beta(s)$ defined in Equation (\ref{eq:R_q}) we write $A^R_{f,g'}$ as 
\begin{align*}
	(A^R_{f,g'}\xi)^\delta(q)&= \sum_\beta\langle g', R^{\delta}_{\beta}(q) f\rangle \xi^\beta(q) \,\text.
\end{align*}

It is self-adjoint for all $f\in H, g'\in H'$ if and only if
$\< g', R^\delta_\beta(q) f\>=
	\< R^\beta_\delta(q) f,g'\>$ for all $f\in H$ and $g'\in H'$, which is 
by Lemma \ref{lem:LW} equivalent to that
$\Delta_H^{-\ima s}R^{\delta}_\beta(q)\Delta_H^{\ima s}$ extends to
a bounded weakly continuous map on the strip $\RR+\ima[0,1/2]$
with boundary value $J_HR_{\d}^\b(q)J_H$ for $s=\ima/2$.
Like in \cite{LW11} this is equivalent to $\und{R}(\slot-q)$ being a bounded analytic matrix
valued function on 
$\RR+\ima(0,\pi)$ with 
boundary values 
$\und{R}^{\alpha\beta}_{\gamma\delta}(q+\ima\pi)=\overline{\und{R}^{\bar\alpha \delta}_{\bar\gamma\beta}(q)}$
almost everywhere,
which is by $\und{S}(q)^\ast=\und{S}(-q)$ equivalent to
$\und{S}^{\alpha\beta}_{\gamma\delta}(\ima\pi - q)=\und{S}^{\bar\gamma\alpha}_{\delta\bar\beta}(q)$.
\end{proof}

\subsubsection{Two-particle left-right scattering matrices}
In this section we give an operator definition for two-particle scattering functions 
which describe the scattering behavior of a left and right moving particle 
in the sense of Fock space excitations.  

Bernard remarked that, for the left-right scattering, two of the conditions
can be combined and thus weakened \cite{Bernard92}. The following is our precise
rendition in terms of standard subspaces.
\begin{definition}
	\label{def:SLRabstract}
Given two standard pairs $(H_\pm,T_\pm)$ on $\H_\pm$, respectively, and operators
 $R^\pm\in\S(H_\pm,T_\pm)$,
we denote by $\S(R^+,R^-)\equiv \S(R^+,H_+,T_+;R^-,H_-,T_-)$ the set of 
all $S\in \U(\H_+\otimes\H_-)$ fulfilling
\begin{enumerate}[ {(}1{)} ]
    \item \emph{Boost covariance:} $[S, \Delta_{H_+}^{\ima t} \otimes \Delta_{H_-}^{-\ima t}] = 0$. 
    \item \emph{Translation covariance:} $[S, T_+(t_+)\otimes T_-(t_-)] = 0$ 
	for all $t_+,t_-\in\RR$.
    \item \emph{Left mixed Yang-Baxter equation:} $R^+_{12}S_{13}S_{23} = S_{23}S_{13} R^+_{12}$ on $\H_+\otimes \H_+\otimes \H_-$.
     \item \emph{Right mixed Yang-Baxter equation:} $R^-_{23}S_{12}S_{13} = S_{13}S_{12}R^-_{23}$ on $\H_+\otimes\H_-\otimes \H_-$.
     \item \emph{Left locality:} $\langle g'\otimes \eta, S(f\otimes \xi)\rangle = 
 	\langle f\otimes \eta, S^\ast (g'\otimes \xi)\rangle$ for 
         all $f\in H_+$, $g'\in H_+'$ and $\xi,\eta\in\H_-$.
     \item \emph{Right locality:} $\langle \eta\otimes g', S(\xi\otimes f)\rangle = 
 	\langle \eta\otimes f, S^\ast (\xi\otimes g')\rangle$ for 
         all $f\in H_-$, $g'\in H_-'$ and $\xi,\eta\in\H_+$.
\end{enumerate}
\end{definition}

Using the physicists' notation, we will define the operator
\[
\xi\mapsto \< g'|_1 S(f\otimes \xi) \equiv
\sum_k \langle g'\otimes e^-_k,S(f\otimes \xi)\rangle\cdot e^-_k
\]
on $\H_-$, where $\{e_k\}$ is an orthonormal basis of $\H_-$ 
and analogously for ``bra'' on the second component.
Left/right locality is with this notation equivalent to self-adjointness of 
the operators $A^\pm_{f,J_{H_\pm}g} \in \B(\H_\mp)$ for all $f,g \in H_\pm$, respectively,
where $A^\pm_{f,J_{H_\pm}g}$ is defined by 
$\xi \mapsto \langle J_{H_+} g|_1 S(f\otimes \xi)$ and
$\xi \mapsto \langle J_{H_-} g|_2 S(\xi\otimes f)$, respectively.

We use the same parametrization as before for the standard pairs $(H_\pm,T_\pm)$.
The fact $[S, \Delta_{H_+}^{\ima t} \otimes \Delta_{H_-}^{-\ima t}] = 0$
and $[S, T_+(t_+)\otimes T_-(t_-)]= 0$
enables us to make for $S\in\S(R^+,R^-)$ the ansatz $S=\und{S}(Q_1+Q_2)$,~i.e.~
\begin{align}
	\label{eq:ansatz2}
	(Sf)^{\alpha\beta}(q_1,q_2) 
	=\und{S}^{\alpha\beta}_{\gamma\delta}(q_1+q_2) f^{\gamma\delta}(q_1,q_2)
\end{align}
where by abuse of notation $\und{S}(\slot)=(\und{S}^{\alpha\beta}_{\gamma\delta}(\slot))$ is 
a matrix valued function. 

The operators $S\in\S(R^+,R^-)$ are characterized as follows:
\begin{proposition}\label{prop:Smatrix}
Let $S\in\S(R^+,R^-)$     then $S$ comes from a matrix valued function $q\mapsto \und{S}(q)=(\und{S}^{\alpha\beta}_{\gamma\delta})(q)$ (using the above parametrization) fulfilling
    \begin{enumerate}[ {(}1{)} ]
        \item \emph{Unitarity:} $\und{S}(q)^\ast=\und{S}(q)^{-1}$
		for almost all $q\in\RR$.
        \item \emph{Left mixed Yang--Baxter identity:} For almost all $q,q'\in \RR$ following holds:
		\begin{align*}
			\und{R}^+(q-q')_{12}\und{S}(q)_{13}\und{S}(q')_{23} &= \und{S}(q')_{23}\und{S}(q)_{13} \und{R}^+(q-q')_{12}
			\,,
			\\
			\text{i.e.} \;\; \sum_{\alpha'\beta'\gamma'}{\und{R}^+}^{\alpha\beta}_{\alpha'\beta'}(q-q')
			\und{S}^{\alpha'\gamma}_{\alpha''\gamma'}(q)
			\und{S}^{\beta'\gamma'}_{\beta''\gamma''}(q') 
			&= \sum_{\alpha'\beta'\gamma'} \und{S}^{\beta\gamma}_{\beta'\gamma'}(q')
			\und{S}^{\alpha\gamma'}_{\alpha'\gamma''}(q)
			{\und{R}^+}^{\alpha'\beta'}_{\alpha''\beta''}(q-q') \,\text.
		\end{align*}
        \item \emph{Right mixed Yang--Baxter identity:} For almost all $q,q'\in \RR$ the following holds:
		\begin{align*}
			\und{R}^-(q-q')_{23}\cdot \und{S}(q)_{12}\cdot \und{S}(q')_{13} &=
			\und{S}(q')_{13}\cdot \und{S}(q)_{12}\cdot \und{R}^-(q-q')_{23}
			\,,
			\\
			\text{i.e.} \;\;\sum_{\alpha'\beta'\gamma'}  {\und{R}^-}^{\beta\gamma}_{\beta'\gamma'}(q-q')
			\und{S}^{\alpha\beta'}_{\alpha'\beta''}(q)\und{S}^{\alpha'\gamma'}_{\alpha''\gamma''}(q') &= \sum_{\alpha'\beta'\gamma'}\und{S}^{\alpha\gamma}_{\alpha'\gamma'}(q') 
			\und{S}^{\alpha'\beta}_{\alpha''\beta'}(q){\und{R}^-}^{\beta'\gamma'}_{\beta''\gamma''}(q-q') \,\text.
		\end{align*}
        \item \emph{Analyticity:} $q\mapsto \und{S}(q)$ is boundary value of a bounded analytic function on
		$\RR+\ima (0,\pi)$.
 	\item  \emph{Mixed unitary-crossing relation:} 
 	$\und{S}^{\alpha\beta}_{\gamma\delta} (q+\ima \pi)  = 
             \overline{\und{S}^{\bar\alpha \delta}_{\bar\gamma \beta}(q)}= 
             \overline{\und{S}^{\gamma\bar\beta}_{\alpha\bar\delta}(q)}$ holds.
    \end{enumerate}
\end{proposition}
\begin{proof}
The above ansatz by a matrix-valued function is the most general ansatz
fulfilling $[S, \Delta_{H_+}^{\ima t} \otimes \Delta_{H_-}^{-\ima t}] = 0$
and $[S, T_+(t_+)\otimes T_-(t_-)] = 0$. Then the two notions of unitarity and Yang--Baxter 
identities can be checked to be pairwise equivalent. The proof that left and right locality 
are equivalent to the analyticity and mixed unitary-crossing relation is
completely analogous to the proof of Proposition \ref{prop:Rmatrix}.

 Namely with $W^\delta_\beta(s):=\und{S}^{\slot\delta}_{\slot\beta}(Q_++s)$ 
 left locality is equivalent to 
 $\langle g',W^\delta_\beta(q) f\rangle = \langle W^\beta_\delta(q)f,g'\rangle$ for all $f\in H_+$ and $g'\in H_+'$, 
 which is
 equivalent to $\Delta_{H_+}^{-\ima s}W^\delta_\beta(q)\Delta_{H_+}^{\ima s}$ extending to a bounded weakly
 continuous
 map on the strip $\RR+\ima[0,1/2]$ with boundary value $J_{H_+}W_\delta^\beta(q)J_{H_+}$ for $s=\ima/2$.
 Similarly, with $V^\gamma_\alpha(s):=\und{S}^{\gamma\slot}_{\alpha\slot}(s+Q_-)$ 
 right locality is equivalent to 
 $\langle g',V^\gamma_\alpha(q) f\rangle = \langle V^\alpha_\gamma(q)f,g'\rangle$ for all $f\in H_-$ and $g'\in H_-'$,
 which is equivalent to 
 $\Delta_{H_-}^{-\ima s}V^\gamma_\alpha(q)\Delta_{H_-}^{\ima s}$
 extending to a bounded weakly
 continuous
 map on the strip $\RR+\ima[0,1/2]$ with boundary value $J_{H_-}V_\gamma^\alpha(q)J_{H_-}$ for $s=\ima/2$.
 
 So left and right locality is equivalent to
  $(\und{S}^{\alpha\beta}_{\gamma\delta}(q))$ being a 
 bounded analytic matrix valued function on 
 $\RR+\ima(0,\pi)$ with 
 boundary values 
 $\und{S}^{\alpha\beta}_{\gamma\delta}(q+\ima\pi)=\overline{\und{S}^{\gamma \bar \beta}_{\alpha\bar\delta}(q)}$
 and 
 $\und{S}^{\alpha\beta}_{\gamma\delta}(q+\ima\pi)=\overline{\und{S}^{\bar\alpha \delta}_{\bar\gamma\beta}(q)}$, respectively, almost everywhere.

\end{proof}

\subsubsection{Examples}\label{examples}
One can see that the conditions in our Proposition \ref{prop:Rmatrix} and
\cite[Definition 2.1]{LS12} are essentially the same: the mass parameters and 
the global gauge action can be added by hand. They assume continuity at the
boundary, but it is clear from the proof that their proof works with non-continuous
boundary values.

Hence, as for $\S(H_+, T_+)$, we have the same set of examples as \cite{LS12}.
We point out that the S-matrix of the $\mathrm O(N)$ $\s$-models
satisfies our conditions, where $H_+$ has multiplicity $N$ and
$\und{S}^{\a\a'}_{\b\b'}(q) = \s_1(q)\d^{\a}_{\a'}\d^{\b}_{\b'} + \s_2(q)\d^{\a}_{\b'}\d^{\a'}_{\b} + \s_3(q)\d^{\a}_{\b}\d^{\a'}_{\b'}$
where $\s_i$ are certain analytic functions on the strip $\RR+\ima(0,\pi)$
(see \cite{AAR01, LS12} for detail)
and $\d$ is the Kronecker Delta.

As for left-right scattering, we present a class of examples. The $\mathrm O(N)$ $\s$-models
can be used to construct examples of this class.
Let us take $R \in \S(H, T)$ and assume that for the corresponding matrix-valued function $\und{R}$ it holds that
$\und{R}^{\a\b}_{\g\d}(q) = \und{R}^{\b\a}_{\d\g}(q)$.
By introducing the component flip operator $(\und{F}\xi)^{\alpha\beta}(q_1,q_2) = \xi^{\beta\alpha}(q_1,q_2)$,
this is equivalent to $\und{F}\und{R}(q)\und{F} = \und{R}(q)$.
Let us say in this case that $R$ satisfies the {\bf component flip symmetry} (this should be distinguished from
the canonical flip which appears below).
It is clear that the S-matrices of the $\mathrm O(N)$ $\s$-models satisfy this.
We claim that $R$ itself can play the role of the left-left, right-right and left-right scatterings.

\begin{proposition}\label{pr:flip}
 If $R \in \S(H,T)$ and satisfies the flip symmetry, then $\breve S = \und{R}(Q_+\otimes \1 + \1\otimes Q_+) \in \S(R,R)$,
 where $\und{R}(\cdot)$ is the to $R$ corresponding operator-valued function,
 namely $\breve S$ is a left-right scattering for the pair $(R,R)$. 
\end{proposition}
\begin{proof}
 From Proposition \ref{prop:Rmatrix} we know that the matrix-valued function
 $\und{S}^{\a\b}_{\g\d} := \und{R}^{\b\a}_{\g\d}$ satisfies the conditions listed there
 and the necessary properties of $\breve S=\und{\breve S}(Q_+\otimes \1 + \1\otimes Q_+)$ in Proposition \ref{prop:Smatrix} can be read off:
 Unitarity is trivial. Since $\breve S$ is defined through the same function
 $\und{R}$, the left mixed Yang-Baxter equations follow trivially from the Yang-Baxter
 equation for $\und{R}$. Note that Proposition \ref{prop:Rmatrix} is written in $\und{S}$ and must
 be translated in $\und{R}$: namely,
\begin{align*}
 \und{R}(q)_{12}\und{R}(q+q')_{13}\und{R}(q')_{23}&=
 \und{R}(q')_{23}\und{R}(q+q')_{13}\und{R}(q)_{12}\,.
\end{align*}
 The right mixed Yang-Baxter equation can be obtained by applying the component flip $\und{F}_{23}$ from the both sides
 to the left mixed Yang-Baxter equation and by using the component flip symmetry of $\und{R}(q)$.
 Analyticity for $\und{\breve S}$ is exactly the analyticity of $\und{R}$.
 Finally, the mixed unitary-crossing relation can be shown as follows:
\[
\und{\breve S}^{\a\b}_{\g\d}(q+\ima\pi) = \und{R}^{\a\b}_{\g\d}(q+\ima\pi) = \und{S}^{\b\a}_{\g\d}(q+\ima\pi) = \und{S}^{\bar\g\b}_{\d\bar\a}(-q) = \overline{\und{S}^{\d\bar\a}_{\bar\g\b}(q)}
= \overline{\und{R}^{\bar\a\d}_{\bar\g\b}(q)} = \overline{\und{\breve S}^{\bar\a\d}_{\bar\g\b}(q)}
\]
where we used the definition of $\und{\breve S}$, the definition of $\und{S}$, the crossing symmetry for $\und{S}$,
Hermitian analyticity of $\und{S}$ and the definitions of $\und{S}$ and $\und{\breve S}$ in this order. This is the first
of the Mixed unitary-crossing relation. The second relation is obtained by applying the component flip symmetry
to the both sides of the first relation and replacing the labels as $\a \leftrightarrow \b, \g \leftrightarrow \d$.
\end{proof}

Hence we obtain a concrete family of left-right scattering operators out of $\mathrm O(N)$ $\s$-models.
We do not know whether there are Lagrangians for our new S-matrices.
We will construct corresponding massless Borchers triples in Section \ref{massless-wedgelocality}
and massive Borchers triples in Section \ref{massive}.
This in turn gives again another family of left-left scattering.
In order to repeat this procedure, it is necessary that the starting $\und{R}$ satisfies
further symmetry $\und{R}(q) = \und{R}(\ima\pi-q)$. We do not know any such example except
constant matrices or scalar case \cite{Tanimoto13}.

\subsection{Second quantization of standard pairs}
\label{sec:2ndQBP}
\subsubsection{\texorpdfstring{$R$}{R}-symmetric Fock space}\label{R-symmetric}
\begin{proposition}
  Let $\H$ be a Hilbert space and 
 $F\equiv F_{12}:\H\otimes \H \longrightarrow \H\otimes \H$ the canonical 
 flip operator given by $F(\xi_1\otimes \xi_2) = \xi_2\otimes \xi_1$. 
Then there is a one-to-one correspondence between
\begin{enumerate}
\item unitary operators $R$ on $\H\otimes\H$ satisfying 
    $R_{21}\equiv F R F = R^\ast$ and the Yang-Baxter equation
    \begin{align*}
 	R_{12}R_{13}R_{23} &= R_{23}R_{13}R_{12}\,\text,
    \end{align*}
\item unitary involutions (i.e.~self-adjoint unitary operators) $\Phi$ on 
    $\H\otimes\H$, such that the braiding relation
    \begin{align*}
	\Phi_{12}\Phi_{23}\Phi_{12} &=\Phi_{23}\Phi_{12}\Phi_{23}
    \end{align*}
    holds, and
 \item families $(D_n:\mathfrak{S}_n \to \U(\H^{\otimes n}))_{n=2,3,\ldots}$ of 
     unitary representations of the symmetric group compatible with all
     inclusions of $\mathfrak{S}_n\subset \mathfrak{S}_m$ ($m>n$) in the following way. 
     Let $\{i_1,\cdots,i_m\}\cup\{i_{m+1},\ldots,i_n\}$ be an ordered partition
     of $\{1,\ldots,n\}$
     and let $\iota_{i_1\cdots i_m}$ be
     the inclusion of $\mathfrak{S}_m$ into $\mathfrak{S}_n$ as the subgroup of permutations of
     $\{i_1,\cdots,i_m\}$ (leaving $i_{m+1},\ldots,i_{n}$ invariant), then 
    \begin{align*}
 	D_m(\pi) &= D_n(\pi)_{i_1\cdots i_n} & 
 	    (\pi \in \iota_{i_1\cdots i_m} (\mathfrak{S}_n) \subset \mathfrak{S}_m),
    \end{align*}
    where $D_n(\pi)_{i_1\cdots i_n}$ acts on $i_1\cdots i_n$-th tensor components.
\end{enumerate}
 The correspondence given by $\Phi = FR$ and $D_n$ is defined 
 by $D_n(\tau_j)=\Phi_{j,j+1}$ for $1\leq j \leq n-1$ and $\tau_j$ is the transposition of 
$j\leftrightarrow j+1$.
\end{proposition}
\begin{proof}
Given unitary $R$ with $FRF = R^\ast$, define $\Phi:=FR$ and 
therefore $\Phi^\ast = R^\ast F = F R = \Phi$. 
If on the other hand a unitary involution $\Phi$ is given, by defining $R:=F\Phi$
we get $R^\ast = \Phi F = R^{-1}$ and 
$FRF = FF\Phi F= \Phi F = R^\ast$.  
It is obvious that $F_{12}F_{23}F_{12}  = F_{23}F_{12}F_{23}$ holds. 
For $F$ and $S$ we get the commutation relation
$S_{12}F_{23} = F_{23}S_{13}$ and therefore
\begin{align*}
\Phi_{23}\Phi_{12}\Phi_{23}
 	&= F_{23}R_{23}F_{12}R_{12}F_{23}R_{23}
 	\\&= F_{23}F_{12}F_{23} \circ R_{12} R_{13} R_{23}\,\text,
\\
\Phi_{12}\Phi_{23}\Phi_{12}
 	&= F_{12}R_{12}F_{23}R_{23}F_{12}R_{12}
 	\\&= F_{12}F_{23}F_{12} \circ R_{23} R_{13} R_{12}\text.\end{align*}
From this, the equivalence between the  1.~and 2.~is clear.
 	
 	For $\tau_i$ the transposition of the $i$-th and $(i+1)$-th element,
 we define $D_n(\tau_i)= \Phi_{i,i+1}$, which gives a representation
of 
$$\mathfrak{S}_n=\langle \tau_1,\ldots,\tau_{n-1}: 
 \tau_i\tau_{i+1}\tau_i=\tau_{i+1}\tau_i\tau_{i+1}\text{ and }\tau_i\tau_j=\tau_j\tau_i \text{ for }|i-j|\geq 2\rangle$$
by the properties of $\Phi$.
 Given $\{D_n\}$ we set $\Phi:=D_2(\tau_1)$ and we observe that 
 the family is already fixed by $D_n(\tau_i) =\Phi_{i,i+1}$, because the 
 transpositions generate $\mathfrak{S}_n$.
\end{proof}

For a pair $(\H,R)$ of a Hilbert space and a unitary 
$R\in\U(\H\otimes\H)$ fulfilling
$R_{21}\equiv F R F =R^\ast$ and the Yang--Baxter identity, 
i.e.~Properties (1) and (2) of Definition 
\ref{def:TwoParticleScatteringOperator},
we associated the Fock space $\F_{\H,R}$ given by 
$$\F_{\H,R}=P_R\F_\H^\Sigma,$$ where $P_R$ is the projection
$$
P_R\restriction \H^{\otimes n} = \frac1{n!} \sum_{\sigma \in S^n} D_n(\sigma)
$$
and 
$$\F_\H^\Sigma =\CC\Omega\oplus\bigoplus_{n=1}^\infty \H^{\otimes n}$$ 
is the unsymmetrized Fock space over $\H$. 
For $\|A\|\leq 1$, such that $[A\otimes A, R]=0$ there is an operator 
$\Gamma(A) = \1 \oplus A \oplus (A\otimes A)\oplus \cdots$, which
restricts to $\F_{\H,R}$.

The construction is functiorial, from the
additive (by taking direct sums) category with
\begin{description}
\item[Objects] Pairs $(\H,R)$ of a Hilbert space and a unitary $R\in\U(\H\otimes\H)$ 
fulfilling $R_{21}= F R F =R^\ast$ and the Yang--Baxter relation.
\item[Morphisms] Contractions $A:(\H_1,R_1)\to (\H_2,R_2)$ with $(A\otimes A )R_1=R_2(A\otimes A)$.
\end{description}
to the multiplicative (by taking tensor products) 
category of Hilbert spaces with contractions,
which is given by
\begin{align*}
(\H,R)&\longmapsto \F_{\H,R}\\
A:(\H_1,R_1)\to (\H_2,R_2) &\longmapsto\Gamma(A)= 1\oplus \bigoplus_{n=1}^\infty A^{\otimes n }:\F_{\H_1,R_1}\to \F_{\H_2,R_2}\,\text.
\end{align*}
We note that $\Gamma(A)$ is well defined because 
from $(A\otimes A)R_1 =R_2(A \otimes A)$ it follows that 
$P_{R_1}\Gamma(A)=\Gamma(A)P_{R_2}=P_{R_1}\Gamma(A)P_{R_2}$ where $P_{R_i}$ is 
here the projection from $\F_{\H_i}^\Sigma$ onto $\F_{\H_i,R_i}$.
It preserves adjoints 
$$
\Gamma(A^\ast)=\Gamma (A)^\ast\,\text.
$$
namely they are preserved on the full Fock space and 
$(A\otimes A) R_1 =R_2(A\otimes A)$ is equivalent to 
$(A^\ast \otimes A^\ast) R_2 =R_1(A^\ast \otimes A^\ast)$
due to $R_i^\ast = (R_i)_{21}$.
In particular, $\Gamma(U)$ is unitary if $U$ is unitary. 
There is a natural isomorphism
\begin{align*}
N:\F_{\H_1\oplus\H_2,R_1\oplus R_2} &\cong \F_{\H_1,R_1}\otimes 
\F_{\H_2,R_2}\\
N\Gamma(A_1\oplus A_2)&= \Gamma(A_1)\otimes \Gamma(A_2)N
\,\text.
\end{align*}
For $A$ antilinear with $(A\otimes A)R =R^\ast (A\otimes A)$ we define 
\begin{align*}
\hat \Gamma(A)&= 
A^0\oplus \bigoplus_{n=1}^\infty F_{1\cdots n} A^{\otimes n}\,\text.
\end{align*}
where for an antilinear operator $A$ we define $A^0$ as the complex conjugation on $\CC$
and $F_{1\cdots n} (f_1\otimes \cdots \otimes f_n)= f_n\otimes \cdots \otimes  f_1$.
This is well-defined, namely we have 
\begin{align*}
F_{1\cdots n}A^{\otimes n} D_n(\tau_i) =
F_{1\cdots n} \Phi_{i+1,i} A^{\otimes n}=
\Phi_{n-i,n-i+1}  F_{1\cdots n}  A^{\otimes n}
= D_n(\tau_{n-i})F_{1\cdots n}A^{\otimes n}
\end{align*}
hence $\hat \Gamma(A)P_R=P_R\hat \Gamma(A)$.
This can also be formulated as 
\begin{align*}
\hat \Gamma(A)\restriction P_R\H^{\otimes n}
&=A^{\otimes n} F_{1\cdots n} 
=A^{\otimes n} \prod_{1\leq i<j \leq n} R_{ij}
\,\text,
\end{align*}
where in the product the operators are lexicographically ordered
from left to right (or equivalently from right to left by YBE).
Namely, for $\psi\in \H^{\otimes n}$, the restricted vector $P_R\psi$ is 
$R$-symmetric in the sense that we have
$F_{i,i+1} P_R \psi= R_{i,i+1}\Psi_{i,i+1}P_R \psi = R_{i,i+1}P_R\psi$.
From this one can show that on $\H^{\otimes n}$ it holds that 
$F_{1\cdots n}P_R =\prod_{1\leq i<j \leq n} R_{ij}  P_R$.


\newcommand{\ri}[1]{\left\langle #1 \right|}

\subsubsection{Second quantization on \texorpdfstring{$R$}{R}--symmetric Fock space}
For $f\in \H$ let $b(f)$ be the creation operator 
on the subspace of finite particles of $\F_\H^\Sigma$, given by
$b(f) \xi = \sqrt{n+1} \cdot f\otimes \xi$ for $\xi \in \H^{\otimes n}$.
Then its adjoint is given by
$b(f)^\ast \xi= \sqrt{n} \cdot \ri f_1 \xi$,
namely 
\begin{align*}
&(h_0\otimes \cdots \otimes h_m, b(f) g_1\otimes \cdots \otimes g_n)
\\&\qquad= \delta_{mn} \sqrt{n+1} (h_0\otimes \cdots \otimes h_m,f \otimes  g_1\otimes \cdots \otimes g_n)
\\&\qquad= \delta_{mn}\sqrt{m+1} \cdot\overline{(f,h_0)}(h_1\otimes \cdots \otimes h_m,  g_1\otimes \cdots \otimes g_n)
\\&\qquad=(b(f)^\ast h_0\otimes \cdots \otimes h_m ,g_1\otimes \cdots \otimes g_n)\,\text.
\end{align*}
Let $\D$ be the vectors with finite particle 
number, i.e.~$\Psi\in\F_{\H,R}$ where 
$n$-th component vanishes for sufficiently large $n$.

We define on $\F_{\H,R}$ the compressed operators $a(f)=P_Rb(f) P_R$
and define the Segal type field $\phi(f)=a(f)+a(f)^\ast$ on $\D$ which 
is symmetric. We note that $f\mapsto \phi(f)$ is just real linear.

\begin{lemma}[{c.f.~\cite[Lemma 4.1.3.]{Lechner06}}] 
\label{lem:ParticleBounds}
Let $N$ be the number operator. On $\Psi\in\D$ holds
\begin{align*}
\|a(f)\Psi\|&\leq \|f\|\cdot \|(N+1)^{\frac12} \Psi\|\,\text,&
\|a(f)^\ast \Psi\|&\leq \|f\|\cdot \|N^{\frac12} \Psi\|\,\text. \\
\end{align*}
\end{lemma}
\begin{proof}
On the unsymmetrized Fock space $\F_\H^\Sigma$ with  $N\Psi_n=n\Psi_n$ one checks
$b(g)^\ast b(f)\Psi_n=(g,f) (N+1)\Psi_n$ and gets 
$b(g)^\ast b(f)=(g,f) (N+1)$ on $\D$. Hence
$\|b(f)\Psi\|^2 =\|(N+1)^{\frac12}\Psi\|^2\cdot \|f\|^2$ which
implies $\|b(f)(N+1)^{-\frac12}\|=\|f\|$. But then 
also the adjoint $(N+1)^{-\frac12}b(f)^\ast=b(f)^\ast N^{-\frac12}$ has the same norm.
Then the bounds follow from $a(f)^{\#}=P_Rb(f)^{\#}P_R$.
\end{proof}

\begin{lemma} It holds:
\begin{enumerate}
\item $\phi(f)$ is essentially self-adjoint on $\D$.
\item $f\mapsto \phi(f)$ is strongly continuous on $\D$.
    \item $f\mapsto \ee^{\ima \phi(f)}$ is strongly continuous (where
    $\phi(f)$ here is the self-adjoint extension).
\item Let $U\in\U(\H)$ with $[U\otimes U,R]=0$, then 
$\Gamma(U)\phi(f) \Gamma(U)^\ast = \phi(Uf)$ on $\D$.
\item If $H$ is cyclic then $\Omega$ is cyclic 
    for the polynomial algebra of $\phi(f)$ with $f\in H$.
\end{enumerate}
\end{lemma}
\begin{proof}
We proceed as in \cite[Proposition 4.2.2]{Lechner06}.
For $\Psi_n\in \D$ with $N\Psi_n=n\Psi_n$ we get
with $c_f=2\|f\|$ with the help of the bounds of Lemma \ref{lem:ParticleBounds} the estimate 
$\|\phi(f)\Psi_n\|\leq \sqrt{n+1}\cdot  c_f\cdot \|\Psi_n\|$. 
Iteratively, we get  
$$
\|\phi(f)^k\Psi_n\|\leq \sqrt{(n+1)\cdots (n+k)} c_f^k\|\Psi_n\|
$$
and for every $t>0$ we have
\begin{align*}
    \sum_{k=0}^\infty \frac{\|\phi(f)\Psi_n\|}{k!} t^k\leq
    \|\Psi_n\| \sum_{k=0}^\infty \sqrt{\frac{(n+k)!}{n!}} \frac1{k!}
    \left( c_f\cdot t\right)^k &\leq \infty
    \,\text.
\end{align*}
By Nelson's Theorem \cite[Theorem X.39]{RSI}
$\phi(f)$ is essentially self-adjoint on 
$\D$.

Next we prove the continuity (c.f.~\cite[Theorem X.41]{RSII}). For 
$\psi\in P_R\H^{\otimes k}$ 
and $f_n \to f$ a sequence in $\H$ we get 
\begin{align*}
\|\phi(f_n)\psi-\phi(f)\psi\|
&= \|\phi(f_n-f)\psi\| \leq 2\sqrt{k+1}\,\left\|f_n-f\right\|\, \|\psi\|
\end{align*}
so $\phi(f_n)\psi \to \phi(f)\psi$ and thus $\phi(f_n)$ converges 
strongly to $\phi(f)$ on $\D$. Since $\D$ is a core for $\phi(f)$
and all $\phi(f_n)$'s, 
it holds that $\ee^{\ima t \phi(f_n)} \to \ee^{\ima t\phi(f)}$ strongly.

Let $U\in\U(\H)$ with $[U\otimes U,R]=0$, then $\Gamma(U)$ commutes 
with $P_R$. For $\xi\in\H^{\otimes n}$ we get 
\begin{align*}
\Gamma(U)a(f)\Gamma(U)^\ast \xi &=\sqrt{n+1} U^{\otimes (n+1)} 
P_R(f\otimes U^{\ast \otimes n}\xi)\\
&=\sqrt{n+1} P_R(Uf \otimes \xi)\\
&= a(Uf)\xi
\end{align*}
and $\Gamma(U)a(f)^\ast\Gamma(U^\ast) = (\Gamma(U)a(f)\Gamma(U^\ast) )^\ast
=   a(Uf)^\ast$, hence we obtain 4.

The cyclicity can be shown inductively, namely by applying 
$\phi(f)$ on $\Omega$ one can show that one obtain a total
set in $P_R\H^{\otimes n}$.
\end{proof}

We
define for every real subspace $H\subset \H$ the von Neumann algebra
\begin{align*}
\M_R(H) &= \left\{ \ee^{\ima \phi(f)}: f \in H \right\}'' \subset \B(\F_{\H,R})
\,\text.
\end{align*}
This can be seen as a generalization
of the CCR and CAR algebra.
\begin{proposition}
\label{prop:CSR}
Let $(\H,R)$ like before and $K,H\subset \H$ real subspaces: 
\begin{enumerate} 
\item $K\subset H$ then $\M_R(K)\subset \M_R(H)$,
\item $\M_R(K) =  \M_R(H)$ if 
$\overline K =\overline H$,
\item Let $U\in\U(\H)$ with $[U\otimes U,R]=0$, then 
    $\Gamma(U)\M_R(H) \Gamma(U)^\ast = \M_R(UH)$,
\item If $H$ is cyclic then $\Omega$ is cyclic for $\M_R(H)$.
\end{enumerate}
\end{proposition}
\begin{proof}
The first statement is clear and the second follows from continuity.
The covariance with respect to unitaries with $[U\otimes U,R]=0$
follows from the covariance of $\phi(f)$.
Let $f_1,\cdots,f_n\in H$ and let $E_k(t)$ be the spectral projection
of the self-adjoint operator $\phi(f_k)$ on the 
spectral values $[-t,t]$. Then $F_k(t):=\phi(f_k)E_k(t)\in\M$ for all $t>0$ 
and $F_k(t) \to \phi(f_k)$ strongly on $\D$ and hence
$F_1(t)\cdots F_n(t)\Omega$ converges to $\phi(f_1)\cdots \phi(f_n)\Omega$ 
for $t\to\infty$. The cyclicity of $\Omega$ for $\M$ then follows from the 
cyclicity of $\Omega$ for $\phi$.
\end{proof}

\subsubsection{\texorpdfstring{$R$}{R}--symmetric second quantization of standards pairs and modular 
theory}
In this section we are interested in the construction of 
one-dimensional Borchers triples 
from a standard pair $(H,T^1)$ on $\H$. It turns out that for all
$R\in\S(H,T^1)$ it is possible to construct a one-dimensional Borchers triple
on the ``twisted Fock space'' $\F_{\H,R}$.

Before we turn to the von Neumann algebras we first need commutation 
relation of the Segal field $\phi(f)$ with the ``reflected Segal field''
$J\phi(f)J$. One can think of $\phi(f)$ for $f\in T^1(a)H$ as a field 
localized in a right half-ray $\RR_++a$ and of  $\phi'(g):=J\phi(J_Hg)J$
as a field localized in the left half-ray $\RR_-+b$ for 
$g\in T^1(b)H'$.
\begin{lemma}\label{lm:FieldLocality}
Let $(H,T^1)$ be a standard pair and 
$R\in \S(H,T^1)$ two-particle scattering operator, $\phi(f)$ the 
operator on $\D \subset \F_{\H,R}$ defined above and $J=\hat\Gamma(J_H)$.
Then for $f,g\in H$ the commutator $[J\phi(g)J,\phi(f)]$ vanishes on $\D$.
\end{lemma}
\begin{proof}
Note that $\<h|_1$ and $\<h|_n$, operators on $\F^\Sigma_\H$, preserve
$P_R \H^{\otimes n}$ because
$P_R \H^{\otimes n}$ is characterized by $R$-symmetry (see Section \ref{R-symmetric})
and $\<h|_1$ and $\<h|_n$ do not affect the decomposition of a permutation into transpositions.
For $\Psi_n\in P_R\H^{\otimes n}$ we get
\begin{align*}
J a(f)^\ast J \Psi_n &=  J a(f)^\ast
    F_{1\cdots n}     J_H^{\otimes n} \Psi_n\\
&=\sqrt{n}\cdot F_{1\cdots (n-1)} J_H^{\otimes(n-1)} \ri f_1 
F_{1\cdots n}     J_H^{\otimes n} \Psi_n\\
&=\sqrt{n}\cdot F_{1\cdots (n-1)} J_H^{\otimes(n-1)}        
    F_{1\cdots (n-1)}     J_H^{\otimes (n-1)} \ri {J_H f}_n  \Psi_n\\
&=\sqrt{n} \ri {J_Hf}_n  \Psi_n \,\text.
\end{align*}
Therefore, we have
\begin{align*}
\left[Ja(g)^\ast J,a(f)^\ast\right]\Psi_n &= \sqrt{n}\left(Ja(g)^\ast J \langle f|_1- a(f)\langle J_H g|_n\right)\Psi_n\\
&= \sqrt{(n-1)n} \left( \langle J_H g|_{n-1}\langle f|_1 - \langle f|_1\langle J_H g|_n\right)\Psi_n\\
&= \sqrt{(n-1)n} \left( \langle f|_1 \langle J_H g|_{n}- \langle f|_1\langle J_H g|_n\right)\Psi_n\\
&=0
\end{align*}
and also $[Ja(g)J,a(f)]=-[Ja(g)^\ast J,a(f)^\ast]^\ast =0$ on $\D$.
To calculate the mixed commutator, 
we first note that (c.f\! \cite[Lemma 4.1.2]{Lechner06})
\begin{align*}
	P_{R}\restriction\H\otimes P_{R}\H^{\otimes n} =\frac1{n+1}\sum_{i=1}^{n+1} X_{1i}
\end{align*}
holds, where $X_{11} = \1$ by convention and $X_{1i}:=D_{n+1}(\tau_{i-1}\cdots \tau_1)\equiv\Phi_{i-1,i}\Phi_{i-2,i-1}\cdots \Phi_{12}=F_{i-1,i}\cdots F_{12}R_{1i}R_{1,i-1}\cdots R_{12}$. 
In other words, this amounts to $R$-symmetrizing the first component since the rest is already $R$-symmetric.
Therefore, the creation operator
acts on $\Psi_n\in P_R\H^{\otimes n}$ by 
$a(f)\Psi_n =  \frac{1}{\sqrt{n+1}}\sum_{i=1}^{n+1} X_{1i} (f\otimes \Psi_n)$ and we calculate:
\begin{align*}
Ja(g)^\ast Ja(f)\Psi_n&= Ja(g)^\ast J \frac1{\sqrt{n+1}} \sum_{i=1}^{n+1} X_{1i}\left(f\otimes \Psi_n\right)\\
&=   \sum_{i=1}^{n+1} \langle J_H g |_{n+1} X_{1i}\left(f\otimes \Psi_n\right)\,\text,\\
a(f)Ja(g)^\ast J \Psi_n
&= a(f) \sqrt{n} \langle J_H g|_n\Psi_n\\
&=  \sum_{i=1}^n X_{1i} \left(f\otimes\left( \langle J_Hg|_n\Psi_n\right)\right)\\
&=  \sum_{i=1}^n \langle J_H g|_{n+1} X_{1i} \left(f\otimes \Psi_n\right)\,\text,\\
[Ja(g)^\ast J,a(f)]\Psi_n&= \langle J_H g|_{n+1} X_{1,n+1} \left(f\otimes \Psi_n\right)
\\
&= \ri{J_Hg}_1 R_{1,n+1}R_{1,n}\cdots R_{12} (f\otimes \Psi_n)\,\text.
\end{align*}
Finally, restricted to $P\H^{\otimes n}$ with $\tilde R= R_{1,n+1}R_{1,n}\cdots R_{12}$,
\begin{align*}
[J\phi(g)J,\phi(f)]
    &= [Ja(g)^\ast J,a(f)] + [Ja(g)J,a(f)^\ast]\\
    &= [Ja(g)^\ast J,a(f)] - [Ja(g)^\ast J,a(f)]^\ast\\
    &\equiv A^{\tilde R}_{f,J_H g}-(A^{\tilde R}_{f,J_H g})^\ast\\
    &=0
\end{align*} 
holds for all $f,g\in H$ because of Lemma \ref{lem:WedgeLocality}.
\end{proof}
\begin{proposition}\label{pr:VNLocality}
Let $(H,T^1)$ be a standard pair with finite multiplicity 
on $\H$ and $R\in \S(H,T^1)$, then
for the von Neumann algebra $\M_R(H)=\{\ee^{\ima \phi(f)}:f\in H\}''$ on $\F_{\H,R}$ 
it holds that:
\begin{enumerate}[ {(}1{)} ]
\item \label{item:BorchersProp}
    $T(t) \M_R(H)T(-t)\subset \M_R(H)$ for $t\geq 0$, where $T(t)=\Gamma(T^1(t))$.
\item \label{item:Standardness}
    $\Omega\in\F_{\H,R}$ is cyclic and separating for $\M_R(H)$.
\item $\Delta^{\ima t}_{(\M_R(H),\Omega)} = \Gamma(\Delta_{H}^{\ima t})$
    and $J_{(\M_R(H),\Omega)}=\hat\Gamma(J_{H})$.
\item $\Omega$ is up to phase unique translation invariant vector in $\F_{\H,R}$.
\end{enumerate}
\end{proposition}
\begin{proof}
(\ref{item:BorchersProp}): This follows from the inclusion of one-particle spaces.

(\ref{item:Standardness}): 
We define $\M_2:=\{\ee^{\ima J\phi(f)J}: f\in H\}$.
Analogously to the case of $\M=\M_R(H)$, it can be shown that $\Omega$ is cyclic for $\M_2$,
so that $\Omega$ is separating for $\M$ 
can be shown by proving $[\M,\M_2]=\{0\}$.

To show that $\M$ and $\M_2$ commute we need to 
use energy bounds. Let $P_0=\dd\Gamma(P_1+1/P_1)\geq 2$ 
with domain $\D_0$ be
the generator of $\Gamma(\ee^{\ima  t (P_1+1/P_1)})$.
We get $P_0 \geq 2N$.
We will see in Section \ref{projection} (only for the irreducible case, but
reducible cases are just parallel) that $P_1$ and $1/P_1$ can be
identified with the generators of positive and negative lightlike
translations in a massive representation. Hence $P_1 + 1/P_1$ is
the generator of the timelike translations.
Real Schwartz test functions with support in $W_\R$ are mapped densely into $H$
as we will see in Section \ref{projection}.
We get bounds from the proof of Lemma \ref{lem:ParticleBounds} and because 
the multiplicity is finite, it holds that 
$\|(1+P_0)^{-\frac12}\phi(f)\| < \infty$ on $\D_0$
and similar for the commutator $[P_0,\phi(h)] = \phi(\partial_0 h)$,
where $h$ is a test function with support in $W_\R$, $\partial_0 h$ is the timelike derivative
and $\phi(h)$ is defined through the mapping mentioned above,
and for $J\phi(h)J, [P_0, J\phi(h)J]$ (see also the argument in \cite[Proposition 3.1]{BL04}).
By the commutator theorem \cite{DF77} one can conclude 
that $\ee^{\ima \phi(h)}$ and $\ee^{\ima J\phi(g)J}$ 
commute for all such $h,g$ which by continuity 
implies that $\M$ and $\M_2$ commute.

The property of the modular operators (3) is proved as in
\cite[Proposition 3.1]{BL04} and $\Omega$ is the unique translation invariant vector,
because we assume that standard pairs are non-degenerate.
\end{proof}

\begin{corollary} For each standard pair (with finite multiplicity in the reducible case)
$(H,T^1)$ and $R\in\S(H,T^1)$ there exists a one-dimensional Borchers triple $(\M_R(H), T, \Q)$ 
and therefore a half-ray local dilation translation covariant net on $\RR$.
\end{corollary}

Special cases of such models were constructed in \cite{BLM11} 
and were proposed as scaling limits of two-dimensional models with 
factorizing S-matrices. We will present a direct relation to massive models 
in two dimensions via a class of Longo-Witten unitaries like in Section \ref{holography},
in other words via the idea of lightfront holography.

\begin{remark}\label{rem:gauge}
Let us note that each $V_1\in \E(H,T^1)$ with $[V_1\otimes V_1,R]=0$ gives
a Longo-Witten unitary $V=\Gamma(V_1)$ for the 
one-dimensional Borchers triple $(\M_R(H),\Gamma(T^1),\Omega)$.
An \textbf{internal symmetry} of a Borchers triple $(\M,T,\Omega)$ is a unitary $U$ 
leaving $\Omega$ invariant with $[U,T(t)]=0$ and $\Ad U(\M)=\M$.
So as a special case, we get internal symmetries by 
second quantization of elements in $\{V_1\in \E(H,T^1): V_1H=H, [V_1\otimes V_1,R]=0\}$. 
Using the characterization of Longo-Witten unitaries in $\E(H,T^1)$
in \cite{LW11}
by matrices of analytic function, we get that these are exactly constant 
matrices in $\U(\CC^n)$ commuting with $R$ in the above sense where $n$ is the multiplicity of $H$.
Therefore, we can associate with $(H,T^1,R)$ a compact group $G\subset \mathrm U(n)$ acting 
by internal symmetries.
\end{remark}

\begin{remark} 
	\label{rem:masses}
	Let us define an operator $M$ as 
	$(Mf)^\alpha(q) = m^\alpha f^\alpha(q)$ with 
	constants $m^\alpha=m^{\bar\alpha}>0$ and define 
	${T^1}'(t)=\ee^{\ima tM^2P^{-1}}$, i.e.\ 
	$({T^1}'(t)f)^\alpha(q)= \ee^{\ima t (m^\alpha)^2\ee^{-q}} f^\alpha(q)$.
	As in \cite{LW11} we get ${T^1}'(t)\in\E(H,T^1)$.
	 We want that 
	$[{T^1}'(t)\otimes {T^1}'(t),R]=0$, so that $T'(t)=\Gamma({T^1}'(t))$ is well-defined and therefore defines a Longo-Witten unitary for $t\leq0$.
	In the notation with matrix-valued functions, this is equivalent to $\ee^{-q_1}m_\g^2 + \ee^{-q_2}m_\d^2 = \ee^{-q_1}m_\a^2 + \ee^{-q_2}n_\b^2$
	if $R^{\a\b}_{\g\d}(q_1-q_2) \neq 0$ and it is further equivalent to that
	$m^\a\neq m^\g$ implies that $R^{\a\bullet}_{\g\bullet}(q)=0$ and
	$m^\b\neq m^\d$ implies that $R^{\bullet\b}_{\bullet\d}(q)=0$ for almost all $q\in\RR$, respectively.
	
	As in Remark \ref{rem:gauge} we can associate a compact group $G$ with $(H,T^1,{T^1}',R)$
	by asking besides $[T^1,V_1]=0$ and $V_1H=H$ that also $[{T^1}',V_1]=0$ holds.
\end{remark}
\subsection{Construction of massless wedge-local models from scattering operators}\label{massless-wedgelocality}

Given two standard pairs $(H_\pm,T^1_\pm)$ on $\H_\pm$, respectively,
and two operators $R^\pm \in\S(H_\pm,T^1_\pm)$ we obtain two one-dimensional Borchers triples
$(\M_\pm,T_\pm,\Omega_\pm)$ by the construction of Section \ref{sec:2ndQBP}.

We show that every $S\in\S(R^+,R^-)$ gives rise to a wave-scattering matrix
$\tilde S$ as in Proposition \ref{pr:massless-structure}.

Let us define the operator $\tilde S=\bigoplus_{m,n}  S^{(m,n)}$ on full Fock space 
$\F^\Sigma_{\H_+}\otimes\F^\Sigma_{\H_-}$ by 
\begin{align*}
	\B(\CC\Omega_+ \otimes \H_-^{\otimes n})\ni S^{(0,n)} &= \1\\
	\B(\H_+^{\otimes m}\otimes \CC\Omega_-)\ni S^{(m,0)} &= \1\\
	\B(\H_+^{\otimes m}\otimes \H_-^{\otimes n})\ni S^{(m,n)}	
	&=S_{1|1}S_{2|1}\cdots S_{m|1}S_{1|2}\cdots S_{m|n}
\end{align*}
where we denote for $1\leq i\leq m$ and $1\leq j\leq n$ by $S_{i|j}\equiv
S^{m|n}_{i|j}$ the operator on $\H_+^{\otimes m}\otimes \H_-^{\otimes n}$
given by $S_{i,j+m}$ (we omit $m|n$ when no confusion arises).
We will use notation as $\<f|_{1|}$, $R^+_{ij|}$ and $R^-_{|ij}$ as well.
Namely if one side of $|$ is empty, then the operator acts trivially on
that side.
\begin{lemma}
	\label{lem:mixedYBE} 
	Let $(H_\pm,T_\pm^1)$ be two standard pairs on $\H_\pm$ respectively and $R^\pm \in \S(\H_\pm,T_\pm^1)$.
	Given an operator $S$ 
	fulfilling the properties (1) and (2) of Definition 
	\ref{def:SLRabstract} and let $\tilde S$ be defined as above.
	Then the following hold.
	\begin{itemize}
		\item $[\tilde S, P_{R^+}\otimes \1_{\F_{\H_-}^\Sigma}]=0$ if and only if the left YBE holds;
		\item $[\tilde S, \1_{\F_{\H_+}^\Sigma}\otimes P_{R^-}]=0$ if and only if the right YBE holds.
	
	\end{itemize}

\end{lemma}
\begin{proof} 
Fix $m\geq 2$. The left mixed YBE 
$R^+_{12}S_{13}S_{23}=S_{23}S_{13}R^+_{12}$ implies on $\H_+^{\otimes m}
\otimes \H_-^{\otimes n}$ the equality
$\Phi^+_{i,i+1} S_{i|\bullet}S_{i+1|\bullet}=S_{i|\bullet}S_{i+1|\bullet}\Phi^+_{i,i+1}$ for $1\leq i\leq m-1$, where $\Phi^+_{ij}=F_{ij}R^+_{ij}$. 
Furthermore $[S_{k|\bullet},\Phi_{i,i+1}]=0$ holds trivially for $1\leq k\leq m$ with $k\notin \{i,i+1\}$.
Because the $P^{(m)}_{R^+}:=P_{R^+}\restriction \H_+^{\otimes m}$ is given by the linear combination
of products of $\Phi_{i,i+1}$ and because $S_{i+1|\bullet}$ always appears next to
$S_{i|\bullet}$ in the definition of $S^{(m,n)}$,
we can conclude that $S^{(m,n)}$ commutes with $P^{(m)}_{R^+}\otimes \1_{\H_-^{\otimes n}}$ for all $n$.
For the second statement, note that we can rewrite $S^{(m,n)}$ as follows:
\[
 S^{(m,n)} =S_{1|1}S_{1|2}\cdots\cdots S_{1|n}S_{2|1}\cdots S_{m|n},
\]
because one has only to exchange the orders of elements which are acting on different tensor components.
Now one proves from the right mixed YBE that $\1_{\H_+^{\otimes m}}\otimes P^{(n)}_{R^-}$ commutes 
with $S^{(m,n)}$ for $m\geq1,n\geq 2$.

The converse holds because the left and right mixed YBE are equivalent to the commutation of 
$\tilde S$ with $P_{R^+}\otimes  \1_{\F_{\H_-}}$ and $\1_{\F_{\H_+}}\otimes P_{R^-}$ restricted to 
$\H_+\otimes \H_+\otimes \H_-$ and $\H_+\otimes\H_-\otimes\H_-$, respectively.
\end{proof}
Therefore, $\tilde S$ canonically restricts to an operator on $\F_{\H_+,R^+}\otimes \F_{\H_-,R^-}$
if and only if the left and right YBE are fulfilled. By abuse of notation we denote the 
restricted operator also by $\tilde S$.

\begin{lemma}
	\label{lem:10.3.3}
	Let $(H_\pm,T_\pm^1)$ be two standard pairs on $\H_\pm$ respectively and $R^\pm \in \S(\H_\pm,T_\pm^1)$.
	Given an operator $S$ 
	fulfilling the properties (1)-(4) of Definition 
	\ref{def:SLRabstract}.

	If left locality holds, then for 
	$\tilde R^+ = R^+_{1,m+1|}\cdots R^+_{12|}S_{1|1}\cdots S_{1|n}$ and 
	the operator  
	$A^{\tilde R^+}_{f,J_{H_+}g}:\Psi_m\otimes \Phi_n\mapsto \langle J_{H_+} g|_1
	 \tilde R^+ (g\otimes \Psi_m\otimes \Phi_n)$
	on $\H_{+,m}\otimes \H_{-,n}$ is self-adjoint for all $f,g \in H_+$.
	
	If right locality holds, then 
	$\tilde R^- = R^-_{|1,n+1}\cdots R^-_{|12}S_{1|1}\cdots S_{m|1}$ and 
	$B^{\tilde R^+}_{f,J_{H_-}g}:\Psi_m\otimes \Phi_n\mapsto 
	\langle J_{H_-} g|_{|1} 
	\tilde R^- (\Psi_m\otimes f\otimes \Phi_n)$
	on $\H_{+,m}\otimes \H_{-,n}$ is self-adjoint for all $f,g \in H_-$.
\end{lemma}
\begin{proof}
The proof is analogous to the proof of Lemma \ref{lem:WedgeLocality}. 
For example for the left case we write
\begin{align*}
	R^+&=\sum_{\beta,\delta}\int (R^+)^\beta_{\delta}(q)\otimes \dd E(q)^\beta_\delta\,\text,&
	S&=\sum_{\beta,\delta}\int W^\beta_{\delta}(q)\otimes \dd E(q)^\beta_\delta
\end{align*}
and from left locality it holds that $W^\beta_{\delta}(q)S_{H_+}\subset S_{H_+} W^\delta_{\beta}(q)$. 
Together with $(R^+)^\beta_{\delta}(q)S_{H_+}\subset S_{H_+} (R^+)^\delta_{\beta}(q)$
it follows like in the above-mentioned proof that $A^{\tilde R^+}$ is self-adjoint.
\end{proof}

\begin{proposition}
	\label{pr:LRlocality}
	Let $(H_\pm,T^1_\pm)$ be two standard pairs, $R^\pm\in \S(H_\pm,T^1_\pm)$
	and the associated one-dimensional Borchers triples $(\M_\pm,T_\pm,\Omega_\pm)$. 
	Let $S$ be an operator fulfilling (1)-(4) of Definition \ref{def:SLRabstract}.
	And let $\tilde S$ be the operator on $\F_{\H_+,R^+}
	\otimes \F_{\H_-,R^-}$ associated with $S$ as above. Then
	\begin{itemize}
		\item $x'\otimes \1$ commutes with $\Ad \tilde S(x\otimes\1 )$ for $x'\in\M_{R^+}'$ and $x\in\M_{R^+}$
	if and only if $S$ satisfies left locality.
		\item $\Ad \tilde S(\1\otimes y)$ commutes with $\1\otimes y'$ for $y\in\M_{R^-}$ and $y'\in\M_{R^-}'$
	if and only if $S$ satisfies right locality.
	\end{itemize}	
\end{proposition}
\begin{proof} 
Let us assume left locality holds. We need to show that  
$[J \phi(g)J\otimes \1, \tilde S(\phi(f)  \otimes \1)\tilde S^\ast]=0$ 
holds, then the statement follows using the energy bounds as in Proposition \ref{pr:VNLocality}.

Let $\Psi_m\otimes \Phi_n \in
  P_{R^+}\H_+^{\otimes m} \otimes P_{R^-} \H_-^{\otimes n}$.
We calculate on the full tensor product space $\H_+^{\otimes m}\otimes \H_-^{\otimes n}$
\begin{align*}
	&\frac1{\sqrt{(n-1)n}}
		[Ja(g)^\ast J\otimes \1,
		\tilde S(a(f)^\ast\otimes \1)\tilde S^\ast](\Psi_m\otimes \Phi_n)\\
	&~=\frac1{\sqrt{n-1}}
		\left(
			(Ja(g)^\ast J\otimes \1 )\langle f|_{1|} S^\ast_{1|n}
			\cdots S_{1|1}^\ast	
			-\tilde S(a(f)^\ast\otimes \1)\tilde S^\ast 
			\langle J_{H_+} g|_{m|}
		\right)(\Psi_m\otimes \Phi_n)\\
	&~=\left(
		\langle J_{H_+}g|_{m-1|}\langle f|_{1|} 
		S^\ast_{1|n}\cdots S_{1|1}^\ast 	-\langle f|_{1|} S^{\ast}_{1|n}\cdots S^\ast_{1|1}
		\langle J_{H_+} g|_{m|}\right)(\Psi_m\otimes \Phi_n)\\
	&~=0,
\end{align*}
where we again used that $\langle h|_\bullet$ preserves
the $R^\pm$-symmetric Fock space (see Lemma \ref{lm:FieldLocality})
as do $\tilde S$ due to Lemma \ref{lem:mixedYBE}. Therefore, we get 
$[Ja(g)^\ast J\otimes \1,\tilde S(a(f)^\ast\otimes \1)\tilde S^\ast]=0$ and 
$[Ja(g)J\otimes \1,\tilde S(a(f)\otimes \1)\tilde S^\ast]
= 0$ by taking the adjoint on $\D$.
To compute mixed commutators we proceed as follows, noting that
$\tilde S$ commutes with $R^+$-symmetrization:
\begin{align*}
	&(Ja(g)^\ast J\otimes \1)
		\tilde S(a(f)\otimes \1)\tilde S^\ast(\Psi_m\otimes \Phi_n) =
	\\
	&\qquad=\frac1{\sqrt{m+1}}(Ja(g)^\ast J\otimes \1)
	\sum_{i=1}^{m+1} X_{1i}S_{1|1}\cdots S_{1|n}
	(f\otimes \Psi_m\otimes \Phi_n)
	\\
	&\qquad= \sum_{i=1}^{m+1}\langle J_{H_+} g|_{m+1|}
		 X_{1i}S_{1|1}\cdots S_{1|n}(f\otimes \Psi_m\otimes \Phi_n)
	\\
	& \tilde S(a(f)\otimes \1)\tilde S^\ast
		(Ja(g)^\ast J\otimes \1)(\Psi_m\otimes \Phi_n) =
	\\
	&\qquad=\sqrt{m}\tilde S(a(f)\otimes \1)\tilde S^\ast 
		\langle J_{H_+}g|_{m|}(\Psi_m\otimes \Phi_n) 
	\\
	&\qquad=\sum_{i=1}^{m} X_{1i}
		S_{1|1}\cdots S_{1|n} \langle J_{H_+} g|_{m+1|}
		(f\otimes \Psi_m\otimes \Phi_n)
	\\ 
	&[Ja(g)^\ast J\otimes \1,\tilde S(a(f)\otimes \1)\tilde S^\ast]
		(\Psi_m\otimes \Phi_n)
	\\
	&\qquad= \langle J_{H_+} g|_{m+1|}
		 X_{1,m+1}S_{1|1}\cdots S_{1|n}(f\otimes \Psi_m\otimes \Phi_n)
	\\
	&\qquad= \langle J_{H_+} g|_{m+1|}
		 F_{m,m+1}\cdots F_{12}R^+_{1,m+1}\cdots R^+_{12} S_{1|1}\cdots S_{1|n}(f\otimes \Psi_m\otimes \Phi_n)
	\\
	&\qquad= \langle J_{H_+} g|_{1|}
		 R^+_{1,m+1}\cdots R^+_{12} S_{1|1}\cdots S_{1|n}(f\otimes \Psi_m\otimes \Phi_n)
	\\
	&\qquad=: A^{\tilde R^+}_{f, J_{H_+}g}(\Psi_m\otimes \Phi_n)
\end{align*}
and it holds on finite particle states that
\begin{align*}
	&[J\phi(g)J\otimes \1,\tilde S(\phi(f)\otimes \1)\tilde S^\ast] \\
	&= 
	[Ja(g)^\ast J\otimes \1,\tilde S(a(f)\otimes \1)\tilde S^\ast]  +
	[Ja(g) J\otimes \1,\tilde S(a(f)^\ast\otimes \1)\tilde S^\ast] \\
	&=[Ja(g)^\ast J\otimes \1,\tilde S(a(f)\otimes \1)\tilde S^\ast]  -
	[Ja(g)^\ast J\otimes \1,\tilde S(a(f)\otimes \1)\tilde S^\ast]^\ast
	\\
	&=0
\end{align*}
where we use that the operator 
$A^{\tilde R^+}_{f,J_{H_+}g}$ 
is self-adjoint for all $f,g \in H_+$ by Lemma \ref{lem:10.3.3}.

For the second statement similar calculation leads to 
\begin{align*}
	&[(\1\otimes Ja(g)^\ast J),
		\tilde S(\1\otimes a(f))\tilde S^\ast](\Psi_n\otimes \Phi_m)
	\\
	&\qquad= \langle J_{H_-} g|_{|1} R^-_{|1,n+1}\cdots R^-_{|1,1} S_{1|1}
		\cdots S_{m|1} (\Psi_m \otimes f \otimes \Phi_n)
	\\
	&\qquad= B_{f, J_{H_-}g}^{\tilde R^-}  (\Psi_m \otimes \Phi_n)
\end{align*}
and the same arguments as above hold.

For the only if part we realize that the commutation of $x'\otimes \1$ with 
$\Ad \tilde S(x\otimes \1)$ implies that $[J \phi(g)J\otimes \1, \tilde S(\phi(f)  \otimes \1)\tilde S^\ast]=0$ on a dense domain. The above calculation 
for the case $m=0$ and $n=1$ shows that left locality holds and right locality is analogous.
\end{proof}

\begin{remark}
Lemma \ref{lem:10.3.3} and Proposition \ref{pr:LRlocality} show that 
Definition \ref{def:SLRabstract} leads to the most general 
form of operators $S$ giving rise to a wave S-matrix
as in Proposition \ref{pr:massless-structure} 
using the Fock space structure. 
But there are known examples where the S-matrix is not of this form.
Namely, for the case $H_\pm$ the irreducible standard pair and $R^\pm= \1$ 
a more general family of wave S-matrix not compatible with the 
Fock structure, has been implicitly constructed 
in \cite{BT12}.
\end{remark}

We summarize the construction.
\begin{proposition}\label{pr:lrnet}
For each pair of standard pairs  $(H_\pm,T^1_\pm)$ with finite multiplicity and 
operators $R^\pm\in\S(H_\pm,T^1_\pm)$, $S\in\S(R^+,R^-)$ there 
is an asymptotically complete (in the sense of waves) Borchers triple
$(\M_{\tilde S}, T, \Omega)$ with wave S-matrix $\tilde S$, defined as in Proposition \ref{pr:massless-structure}.
\end{proposition}

\begin{remark}
 We recall that in \cite{Tanimoto12-2, BT12} we proved the corresponding
 commutation by decomposing the S-matrix into Longo-Witten unitaries.
 In this paper we took a slightly different strategy.
 This was necessary for nondiagonal S-matrix, which is more complicated
 and does not admit a simple decomposition into Longo-Witten unitaries.
 On the other hand, the commutation relation
 we needed is $[x'\otimes \1, \Ad \tilde S(x\otimes \1)] = 0$ and
 it is sufficient that $\Ad \tilde S(x\otimes \1) \in \M\otimes \B(\F_{H_-, R^-})$
 hence on the $\B(\F_{H_-, R^-})$ side one has a greater freedom.
 One has to consider not Longo-Witten endomorphisms of $\M$ but commutation relations on a
 larger space. After this observation one can follow the same line of the proofs in \cite{Tanimoto12-2}.
 
 The connection of these extended commutation relations to nets with boundary \cite{LW11} is unclear.
\end{remark}

We showed in \cite[Section 3]{Tanimoto12-2} that the asymptotic chiral components
are conformal if the two-dimensional Borchers triple is strictly local.
Conversely, in order to construct strictly local Borchers triples, one has
to take strictly local one-dimensional components from the beginning.
The question whether one-dimensional Borchers triples can be strictly local
has been considered in \cite{BLM11}, which largely remains open.

From the bootstrap approach, there have been found form factors of some
local operators in certain massless models \cite{DMS95, MS97}.
However, the existence of form factors by no means implies the existence
of the corresponding Haag-Kastler net.
Indeed, we showed \cite{Tanimoto12-2, Tanimoto13} that in massless models with
a prescribed S-matrix, the strict locality can fail. This should be connected
with the well-known problem of the convergence of form factors, which is
clearly worse in massless cases.
Yet, the possibility that one-dimensional Borchers triples can be strictly
local is a very interesting problem. We will discuss this point later in
Section \ref{holography}.

\section{Massive models from left-right scattering}\label{massive}
In this short section we construct massive Borchers triples.
For a given standard pair $(H_+,T_+^1)$, we define the opposite
standard pair as follows: Let $P_+^1$ be the generator of $T_+^1$.
We put $T_+^{1\prime}(t) = \ee^{\ima t/P_+^1}$ and
$T_+'(t) = \Gamma(T_+^{1\prime})$.
\begin{lemma}\label{lm:opposite}
 The pair $(H_+', T_+^{1\prime})$ is a standard pair.
\end{lemma}
\begin{proof}
A standard pair admits the direct sum decomposition as in \cite{LW11}.
With this decomposition, our claim follows from the result for the irreducible
pair \cite[Theorem 2.6]{LW11}, namely $T_+^{1\prime}(t)H_+ \subset H_+$ for $t \le 0$.

One can use the converse of the one-particle Borchers theorem as well \cite[Theorem 2.2.3]{Longo08}.
\end{proof}
If $R^+ \in \S(H,T^1)$, then $[R^+, T_+^{1\prime}\otimes \1] = [R^+, \1\otimes T_+^{1\prime}] = 0$, since
$T_+^{1\prime}$ is defined by a functional calculus of $T_+^1$.
Hence it is clear that the second quantization $T_+' = \Gamma(T_+^{1\prime})$ restricts to
the $R^+$-symmetrized Fock space $\F_{H_+, R^+}$.

Let us recall that one can construct a Borchers triple $(\M_+, T_+, \Q_+)$ (Section \ref{sec:2ndQBP}).
From Lemma \ref{lm:opposite} it follows that $\Ad T_+'(t)$ preserves $\M_+$ for $t \le 0$.
This is equivalent to that $\Ad T_+'(t)$ preserves $\M_+'$ for $t \ge 0$.
Two representations $T_+$ and $T_+'$ obviously commute, both have the positive generator.
Hence the joint spectrum of the combined representation $T_+(t_+)T_+'(t_-)$ of $\RR^2$ is
contained in $\overline{V_+}$.
Furthermore, if $(t_+, t_-) \in W_\R$, or equivalently if $t_+ \le 0$ and $t_- \ge 0$
(see Figure \ref{fig:wedge-local net} and note an unusual definitions of $t_+, t_-$),
then $\Ad T_+(t_+)T_+'(t_-)(\M_+') \subset \M_+'$.
Namely $(\M_+', T_+T_+', \Q_+)$ is a two-dimensional Borchers triple.

By a parallel reasoning, one sees that $(\M_-, T_-'T_-, \Q_-)$ is a two-dimensional Borchers triple,
where $T_-'$ is constructed analogously, but here $t_+$-lightlike translations are
given by $T_-'$ and $t_-$-translations by $T_-$.

\begin{theorem}
 Let $S \in \S(R^+, R^-)$. Then $(\tilde \M_S, \tilde T, \tilde \Q)$ is a Borchers triple,
 where
 \begin{itemize}
  \item $\tilde \M_S = \M_+'\otimes\1\vee \Ad \tilde S(\1\otimes \M_-)$,
  \item $\tilde T(t_+,t_-) = T_+(t_+)T_+(t_-)'\otimes T_-'(t_+)T_-(t_-)$,
  \item $\tilde \Q = \Omega_+\otimes\Omega_-$.
 \end{itemize}
\end{theorem}
\begin{proof}
The properties for $\tilde T$ and $\tilde \Q$ are obvious.
It follows from the properties of their two-particle components
that $\tilde T$ and $\tilde S$ commute, hence $\Ad \tilde T(t_+,t_-) (\tilde \M) \subset \tilde \M$
for $(t_+,t_-) \in W_\R$.
The cyclicity and separating property of $\tilde \Q$ have been already proven
in Proposition \ref{pr:lrnet}.
\end{proof}

We will see in Section \ref{projection} that if $(H_+, T_+^1)$ is irreducible,
then $T_+(t_+)T_+(t_-)'$ is a massive representation. It follows immediately that
for a reducible pair $(H_+, T_+^1)$ the representation $T_+(t_+)'T_+(t_-)$ is just the
massive representation with the same multiplicity. Accordingly, we can call
$(\tilde \M_S, \tilde T, \tilde \Q)$ a massive Borchers triple.

It can be easily realized that the construction here is a generalization
of \cite[Section 6]{Tanimoto13}. Indeed, the present construction takes
two standard pairs, not only irreducible ones, and promotes them by
$R^\pm$-symmetric second quantization, not only by symmetric or antisymmetric
second quantization. Finally, the operator $S$ is allowed to have matrix-value,
not only scalar. It is also a generalization of \cite[Section 3]{Tanimoto13},
because $S$ can depend on the rapidity. However, here we will not investigate
the strict locality.

One may wonder if the S-matrices from our previous work \cite{BT12} can be
used, which does not preserve the two-particle space.
This does not work, at least straightforwardly,
because it is not clear whether the S-matrix commutes with
the opposite translation $T_+'\otimes \1$.

Finally, we remark that our construction in this section is a special case
of \cite{LS12}. To see this, it is enough to extract a Zamolodchikov-Fadeev
algebra from our algebra.
This can be done exactly as in \cite[Section 6]{Tanimoto13}	
and we omit the proof.
As in \cite{Tanimoto13}, our von Neumann algebra is a tensor product
twisted by $\tilde S$, hence the scattering inside a component remains the same.
We just illustrate how the two-particle S-matrix looks like:
As one sees from the construction, the first component is parity-transformed (c.f.\ Section \ref{holography}),
hence the scattering is determined by $\und{R}^{+\prime}(q) = \und{R}^+(\ima \pi-q)$.
If both the multiplicities of $(H_\pm, T_\pm^1)$ are two, then understanding
the $q$-dependence implicitly, it is given by
{\fontsize{8}{3}
\[
\left(\begin{array}{cccc|cccc|cccc|cccc}
{\und{R}^{+\prime}}^{11}_{11} & {\und{R}^{+\prime}}^{12}_{11} & & &            {\und{R}^{+\prime}}^{21}_{11} & {\und{R}^{+\prime}}^{22}_{11}& & &                                & & & &                           & & & \\
{\und{R}^{+\prime}}^{11}_{21} & {\und{R}^{+\prime}}^{12}_{12} & & &            {\und{R}^{+\prime}}^{21}_{21} & {\und{R}^{+\prime}}^{22}_{21} & & &                                & & & &                           & & & \\
                & & & &              & & & &          \overline{\und{S}^{11}_{11}} & \overline{\und{S}^{11}_{21}} & & &   \overline{\und{S}^{11}_{12}} & \overline{\und{S}^{11}_{22}} & & \\
                & & & &              & & & &          \overline{\und{S}^{12}_{11}} & \overline{\und{S}^{12}_{21}} & & &   \overline{\und{S}^{12}_{12}} & \overline{\und{S}^{12}_{22}} & & \\
\hline
{\und{R}^{+\prime}}^{11}_{12} & {\und{R}^{+\prime}}^{12}_{12} & & &            {\und{R}^{+\prime}}^{21}_{12} & {\und{R}^{+\prime}}^{22}_{12} & & &          &                        & & &   &                       & & \\
{\und{R}^{+\prime}}^{11}_{22} & {\und{R}^{+\prime}}^{12}_{22} & & &            {\und{R}^{+\prime}}^{21}_{22} & {\und{R}^{+\prime}}^{22}_{22} & & &          &                        & & &   &                       & & \\
 &   & & &           &   & & &          \overline{\und{S}^{21}_{11}} & \overline{\und{S}^{21}_{21}} & & &   \overline{\und{S}^{21}_{12}} & \overline{\und{S}^{21}_{22}} & & \\
 &   & & &           &   & & &          \overline{\und{S}^{22}_{11}} & \overline{\und{S}^{22}_{21}} & & &   \overline{\und{S}^{22}_{12}} & \overline{\und{S}^{22}_{22}}          & & \\
\hline
 & & \und{S}^{11}_{11} & \und{S}^{12}_{11} &    & & \und{S}^{21}_{11} & \und{S}^{22}_{11} &    & &   & &          & &   & \\
 & & \und{S}^{11}_{21} & \und{S}^{12}_{21} &    & & \und{S}^{21}_{21} & \und{S}^{22}_{21} &    & &   & &          & &   & \\
 & &                        & &    & &              & &    & & {\und{R}^-}^{11}_{11} & {\und{R}^-}^{12}_{11} &          & & {\und{R}^-}^{21}_{11} & {\und{R}^-}^{22}_{11} \\
 & &                        & &    & &              & &    & & {\und{R}^-}^{11}_{21} & {\und{R}^-}^{12}_{21} &          & & {\und{R}^-}^{21}_{21} & {\und{R}^-}^{22}_{21} \\
\hline
 & & \und{S}^{11}_{12} & \und{S}^{12}_{12}    & & & \und{S}^{21}_{12} & \und{S}^{22}_{12} &   & & &   &          & & & \\
 & & \und{S}^{11}_{22} & \und{S}^{12}_{22}    & & & \und{S}^{21}_{22} & \und{S}^{22}_{22} &   & & &   &          & & & \\
 & & &                       &     & & &                        &   & & {\und{R}^-}^{11}_{12} & {\und{R}^-}^{12}_{12} &          & & {\und{R}^-}^{21}_{12} & {\und{R}^-}^{22}_{12}\\
 & & &                       &     & & &                        &   & & {\und{R}^-}^{11}_{22} & {\und{R}^-}^{12}_{22} &          & & {\und{R}^-}^{21}_{22} & {\und{R}^-}^{22}_{22}
\end{array}\right).
\]}
Using the convention of \cite{LS12} and with an appropriate basis,
an S-matrix of this form could be called {\it block diagonal}.
Of course, such an S-matrix has been already treated in \cite{LS12} in more generality.
The point here is that one can obtain concrete examples from massless
left-right scattering.

\section{Further construction of massive models}\label{holography}
Here we investigate another connection between two-
and one-dimensional Borchers triples. In Sections \ref{massless}, \ref{massive}
our construction has always been carried out on the tensor product Hilbert space.
In this Section we work on a single Hilbert space.

A similar connection has been proposed under the name of algebraic lightfront holography
\cite{Schroer05}. There has been also an effort to reconstruct a full QFT net
from a set of a few von Neumann algebras and some additional structure \cite{Wiesbrock98}
where, however, strict locality remains open.
We present a simple sufficient condition in order to reconstruct a strictly local
Borchers triple out of a conformal net. This sufficient condition turns out to be
hard to satisfy, but we believe that it is of some interest, because techniques to construct models
are rather scarce.

The idea to recover the massive free field from the $\uone$-current through the endomorphisms
associated with the functions $\ee^{\ima t/p}$ is due to Roberto Longo.
Some of the results in this Section have already appeared in the Ph.D. thesis
of the author (M.B.) \cite{Bischoff12}.

\subsection{Holographic projection and reconstruction}
Let $(\M, T, \Omega)$ be a (two-dimensional) Borchers triple.
As we explained in Section \ref{half-line}, $T$ can be restricted to
the lightray $t_+ = 0$, the restriction we denote by $T^+$,
and the triple $(\M, T^+, \Q)$ is a one-dimensional Borchers triple.
We observe that the negative lightlike translation $T^{+\prime}$ is now reinterpreted
as a one-parameter semigroup of Longo-Witten endomorphisms. Indeed, $T^{+\prime}$ obviously commutes with
$T^+$ and $\Ad T^{+\prime}(t_+)$ preserves $\M$ for $t_+ \le 0$.
Furthermore, $T^{+\prime}(\cdot)$ has the positive generator.
These properties of $T^{+\prime}$ are actually very rare if we exclude
the massless asymptotically complete case which we considered in Section \ref{massless}.

Now let us reformulate the situation the other way around.
Let $(\M,T^+,\Omega)$ 
be a one-dimensional Borchers triple and $V(t)$
be a one-parameter semigroup of Longo-Witten endomorphism for $t \le 0$ with positive generator.
Let $T(t_+,t_-) = V(t_+)T^+(t_-)$. By assumption $T^+$ and $V$ commute,
hence $T$ is a representation of $\RR^2$. By the assumed spectral conditions,
$\sp T \subset \overline{V_+}$.
Then we have the following.

\begin{theorem}\label{th:holography}
The triple $(\M, T^+, \Q)$ is a Borchers triple.
If $(\M, T^+, \Q)$ is strictly local, then so is $(\M, T, \Q)$.
\end{theorem}
\begin{proof}
The first statement is clear from the definition.

We assume that $(\M, T^+, \Q)$ is strictly local.
Let $(t_+, t_-) \in W_\R$, in other words $t_+ < 0$ and $t_- > 0$.
One observes that $\Ad T^+(t_-) \circ \Ad V(t_+)(\M) \subset \Ad T^+(t_-)(\M)$.
The intersection in question is
\begin{eqnarray*}
 \M \cap \Ad T(t_+, t_-)(\M') &=& \M \cap \Ad T^+(t_-) \circ \Ad V(t_+)(\M') \\
&\supset& \M \cap \Ad T^+(t_-)(\M')
\end{eqnarray*}
and $\Q$ is cyclic for the right-hand side by assumption.
\end{proof}

As a strictly local Borchers triple corresponds to a Haag-Kastler net,
this Theorem gives a simple construction strategy.
However, as a natural consequence of difficulty in constructing Haag-Kastler nets,
examples of such Longo-Witten endomorphisms seem very rare.

Let us take a closer look at this phenomenon.
We take the Borchers triple $(\M, T^+, \Q)$ associated with the $\uone$-current
net. Among the endomorphisms found by Longo and Witten, the only one-parameter
family with positive generator (negative in their convention \cite{LW11})
is given by the function $\f(p) = \ee^{\ima t/p}$ with $t\leq 0$. 
As we will see, if we take $V_\f = \Gamma(\f(P_1))$, the above prescription gives just the free massive field net,
hence is not very interesting. However, this endomorphisms
is expected not to extend to any extension of the $\uone$-current net,
due to the failure of H\"older continuity of the function $\ee^{\ima t/p}$
at $p = 0$ for $t<0$.
We found another family of such endomorphisms in \cite{BT12}. We will discuss it
in Section \ref{mixing}.

General properties of such endomorphisms have been studied in \cite{Borchers97}.
It is very interesting to find out how to construct more examples of
one-parameter semigroup of Longo-Witten endomorphisms with the semibounded generator,
which would immediately lead to Haag-Kastler nets.

\subsection{Examples}\label{projection}
\subsubsection*{Standard pairs and two-dimensional Wigner representations}\label{wigner}
First we show that from a irreducible standard 
pair we can obtain a representation of the two-dimensional Poincar\'e
group. Everything could be done abstractly by using Borchers commutation 
relations, but we rather give a proof using an explicit representation
to get in contact with models constructed in the literature.

Let $U_m$ be the irreducible positive-energy representation of the 
the two-dimensional proper Poincar\'e group $\P_+$ with mass $m>0$ on a Hilbert space denoted by $\H_m$.
We can identify $\H_m=L^2(\RR,\dd \theta)$ and the action is given by 
\begin{align*}
    p_m(\theta)&= (m\cosh\theta,m\sinh\theta)\\
    (p^0,p^1)\cdot (a^0,a^1) & = p^0a^0-p^1a^1\\
    (U_m(x,\lambda)f)(\theta) &=\ee^{\ima p_m(\theta)\cdot x}\psi(\theta-\lambda)\\
    J_m\psi(\theta)&=\overline{\psi(\theta)}
    \,\text,
\end{align*}
where $J_m=U_m(-I)$ is the anti-unitary representation of $(a^0,a^1)\mapsto(-a^0,-a^1)$.
We remind that we can associate a standard space $H_m(W_\R)$ with the 
right wedge using modular localization \cite{BGL02}, namely 
$H_m(W_\R)=\ker(\1-S_m)$ is the standard space associated with $S_m=J_m\Delta_m^\frac12$,
where $\Delta_m^{\ima t}=U_m(0,-2\pi t)$.

For the irreducible standard pair it is convenient to 
take the restriction to the translation subgroup $\{T_0(t)\}_{t\in\RR}$ 
of the lowest weight $1$ positive energy 
representation of the M\"obius group $\Mob$ on $\H_0$ and the 
standard subspace $H_0=H_0(\RR_+)$ defined again through modular localization \cite{Longo08}.
It can be represented on  $\H_0=L^2(\RR_+,p\dd p)$ by
\begin{align*}
    (T_0(t)f)(p)&=\ee^{\ima t p} f(p)\\
    (\Delta_0^{\ima t})f(p) &= \ee^{-2\pi t}f(\ee^{-2\pi t} p)\\
     J_0f(p)&=\overline{ f(p) }
\end{align*}
such that $(J_0,\Delta_0)$ are the modular objects for $H_0$.
\begin{proposition}
    \label{prop:OneParticleSpaceMassiv}
    Let $(H_0,T_0)$ be the irreducible standard pair and $V_m(s)=\ee^{\ima m^2s/P_0}$, where $T_0(t)=\ee^{\ima t P_0}$.
    Then
    $U_m(a,\lambda)=V_m(\frac12(a^0-a^1))T_0(\frac12(a^0+a^1))\Delta_0^{-\ima\frac\lambda{2\pi}}$
    gives the mass $m$ representation and $H_0$ is identified 
    with $H_m(W_\R)$. 
\end{proposition}
\begin{proof}
    We show using the explicit parametrization. First we note that
    \begin{align*}
        R_m:L^2(\RR_+,p\dd p) & \longrightarrow L^2(\RR,\dd \theta)\\
            f&\longmapsto (\theta\mapsto m\ee^{-\theta}f(m\ee^{-\theta}))
    \end{align*}
    defines a unitary, namely
    \begin{align*}
        (R_mf,R_mg)_{L^2(\RR_+,p\dd p)} &=\int_\RR \overline{R_mf(\theta)}R_mg(\theta) \dd\theta\\
            &= \int_\RR\overline{f(\ee^{-\theta+\ln m})}g(\ee^{-\theta+\ln m})\ee^{-2\theta+2\ln m}\dd\theta\\
            &= \int_\RR\overline{f(\ee^{-\theta})}g(\ee^{-\theta})\ee^{-2\theta}\dd\theta\\
            &= \int_\RR\overline{f(p)}g(p) p \dd p\\
            &=(f,g)_{L^2(\RR_+,p\dd p)}
    \end{align*}
    shows unitarity. Then using 
    \begin{align*}
        \left(V_m(\textstyle\frac12(a^0-a^1))T_0(\textstyle\frac12(a^0+a^1))\Delta_0^{-\ima\frac\lambda{2\pi}}f\right)(p)
        &=\ee^\lambda\ee^{\ima\frac12(a^0+a^1) p+\ima\frac{m^2}2(a^0-a^1)/ p}f(\ee^\lambda p)
    \end{align*}
    we get:
    \begin{align*}
        \left(R_mV_m(\textstyle\frac12(a^0-a^1))T(\textstyle \frac12(a^0+a^1))\Delta^{-\ima\frac\lambda{2\pi}}f\right)(\theta)
        &=m\ee^{-\theta+ \lambda}\ee^{\ima\frac m2(a^0+a^1)\ee^{- \theta}+\ima\frac m2(a^0-a^1)\ee^{\theta}}
        f(m\ee^{-\theta+ \lambda})
        \\
        &=\ee^{\ima p_m(\theta)\cdot a}(R_mf)(\theta-\lambda)\\
        &=U(x,\lambda)(R_mf)(\theta)\,\text,
    \end{align*}
    in particular, one has $R_mT_0(\frac12(a^0+a^1))
    V_m(\frac12(a^0-a^1))\Delta_0^{-\ima\frac\lambda{2\pi}}=U_m(x,\lambda)R_m$.
    $J_\bullet$ acts in both representation 
    by complex conjugation, so it holds also $R_mJ_m=J_0R_m$.
\end{proof}

\subsubsection*{Factorizing S-matrix models}
We exhibited some examples of previously known Borchers triples in
Section \ref{borchers-triples}. The restriction to the lightray gives
one-dimensional Borchers triples as we observed in Section \ref{half-line}.
On the other hand, the scaling limit of the models \cite{Lechner08}
has been investigated and some one-dimensional Borchers triples
(half local quantum fields, in their terminology) have been introduced \cite[Section 4]{BLM11}.
Here we observe that they simply coincide.
As a special case, the lightray-restriction of the massive free net corresponds
to the $\uone$-current net, which we used in \cite{Tanimoto13}.

The one-dimensional Borchers triples in \cite{BLM11} are given as follows:
Let us fix $S_2$. The Hilbert space is the same $S_2$-symmetric
Fock space $\H_{S_2}$ based on the irreducible one-particle space
$L^2(\RR, \dd\theta)$.
The representation $T$ is restricted to the positive lightray,
which acts on the one-particle space as $T(t)(\xi)(\theta) = \ee^{\ima t\ee^\theta}\xi(\theta)$.
For a test function $g$ on $\RR$, the von Neumann algebra is in our notation given by
\begin{gather*}
 \phi_{S_2}(g) = z_{S_2}^\dagger(\hat g^+) + z_{S_2}(J_1\hat g^-), \\
 \N_{S_2} = \{\ee^{\ima\phi_{S_2}(g)}: \supp g \subset \RR_-\}'. 
\end{gather*}
where $\hat f^\pm(\theta) = \pm \ima \ee^\theta \int f(t)\ee^{\ima t\ee^{\pm\theta}}\dd  t = \pm \ima \ee^\theta\tilde f(\pm \theta)$,
where $\tilde f$ is the Fourier transform of $f$.
Note that in our notation $z_{S_2}(\cdot)$ is antilinear, while \cite{BLM11} it is
linear. $T_{S_1}$ and $\Q_{S_2}$ are same as before.

Let us compare this with the von Neumann algebra of the two-dimensional
Borchers triple. It is almost the same:
\begin{gather*}
\M_{S_2} := \{\ee^{\ima\phi_{S_2}(f)}: \supp f \subset W_\L\}', \\
\phi_{S_2}(f) := z_{S_2}^\dagger(f^+) + z_{S_2}(J_mf^-), 
\qquad f^\pm(\theta) = \frac{1}{2\pi}\int d^2af(a)\ee^{\pm \ima p(\theta)\cdot a},
\end{gather*}
Let us consider a function $f(t_+,t_-) = g_1(t_+)g_2(t_-)$.
Then $f^\pm(\theta) = \tilde g_1(-\ee^\theta)\tilde g_2(\ee^\theta)$.
If we take $g_1$ which is the derivative of $g$ above and $g_2$ which approximates the
delta function, it is clear that $f^\pm$ approximate $\hat g^\pm$,
hence we obtain $\N_{S_2} \subset \M_{S_2}$. By the standard argument using
the cyclicity of $\Omega_{S_2}$ and Takesaki's theory (see, e.g.\! the final paragraph of
\cite[Theorem 2.4]{LST12}) one can conclude that $\N_{S_2} = \M_{S_2}$.
Namely the one-dimensional Borchers triples coincide.

Finally, we observe that the case $S_2(\theta)=1$ corresponds to the $\uone$-current net.
The one-particle Hilbert spaces are identified as above and the full spaces are
the symmetric Fock spaces, thus they coincide. Translations are also identified.
In this case one can directly take $\Mr := \{\ee^{\ima\phi_{\mathrm{r}}(f)}: \supp f \subset W_\R\}''$
If one takes $f^+$ as in the previous paragraph where
$f$ is a test function supported in $W_\R$, then as shown in \cite{Lechner03},
$f^+(\theta -\lambda)$ has an analytic continuation in $\RR + \ima(0,\pi)$
and $f^+(\theta -i\pi) = f^-(\theta)$ and it is clear that $J_m f^- = f^+$.
In other words, $f^+ \in \ker (\1- J_m\Delta_m^\frac{1}{2}$).
As the wedge-algebra of the $\uone$-current net is generated
by the exponentiated fields $\A_{\uone}(\RR_+) = \{\ee^{\ima \phi(f^+)}: f^+ \in H(\RR_+)\}''$,
this coincides with $\Mr$ again by Takesaki's theorem.

\subsubsection*{From 1D Borchers pairs}
For a Borchers pair $(H,T^1)$ and $R\in\S(H,T^1)$ and 
a operator $M$ given in our representation by
$(Mf)^\alpha(q)=m^\alpha f^\alpha(q)$ as in Remark \ref{rem:masses}
we can define a massive Borchers triple
$(\M_R(H), TT',\Omega)$ where 
$T'(t)=\Gamma({T^1}'(t))$ and ${T^1}'(t)=\ee^{\ima tM^2P^{-1}}$.
The one-particle spaces can be identified with a direct sum of the spaces 
$\H_{m_\alpha}$ like in Proposition \ref{prop:OneParticleSpaceMassiv}.
Each $\alpha$ corresponds, therefore, to a massive particle with mass
$m_\alpha$.
It is clear that the former example is just a special case 
and one obtains	 in this way the models in \cite{LS12}, namely
from these assumptions about the particle spectrum and the two particle 
scattering operator one can construct the needed data $(H,T^1,{T^1}',R)$.

\subsubsection*{Conjecture on the SU(2)-current algebra}
Zamolodchikov and Zamolodchikov conjectured \cite{ZZ92} that, in our terminology,
the one-dimensional Borchers triples constructed out of the S-matrix of
the $\mathrm{SU}(2)$-symmetric Thirring model is
equivalent to the $\mathrm{SU}(2)$-current algebra, the chiral component
of a conformal field theory. This conjecture, if it turns out to be true,
would imply that the $\mathrm{SU}(2)$-current net admits a one-parameter
semigroup of Longo-Witten endomorphisms with positive generator,
which comes from the negative lightlike translation in the $\mathrm{SU}(2)$-Thirring model.
As remarked before, no such semigroup is so far known for
the $\mathrm{SU}(2)$-current net, hence this would be already new.
Furthermore, as we see in the next Section, if we have two such semigroups,
under suitable technical conditions we can ``mix'' them to obtain
new strictly local Borchers triples, or equivalently Haag-Kastler nets,
which would be a striking consequence.

As far as the authors understand, the conjecture remains open.
Nakayashiki found a quite large family of form factors of the $\mathrm{SU}(2)$-Thirring model
which have the same character as the $\mathrm{SU}(2)$-current algebra at level $1$ \cite{Nakayashiki04}.
However, it is not known whether the current algebra itself is appropriately represented.
As another evidence, it has been revealed that both the $\mathrm{SU}(2)$ Thirring model
and the $\mathrm{SU}(2)$-current algebra admit the same symmetry,
so-called Yangian symmetry \cite{BL93, MacKay05}. Yet the equivalence of the two-models
is unknown.

\subsection{Mixing models by the Trotter formula}\label{mixing}
Here we present a novel idea to construct strictly local Borchers triples.
This has not led to any new example, but the authors expect that there should be concrete situations
where it can apply, as we explain later.

\begin{proposition}\label{pr:trotter}
Let $(\M, T^+, \Q)$ be a one-dimensional Borchers triple and
assume that $\Ad V_1(t), \Ad V_2(t)$, $t\le 0$ are one-parameter Longo-Witten endomorphisms
with positive generators $Q_1, Q_2$.
Furthermore, we assume that $Q_1 + Q_2$ is essentially self-adjoint.
Then $V(t) := \ee^{\ima t(Q_1+Q_2)}$ implement a one-parameter semigroup of Longo-Witten
endomorphism for $t\le 0$ with positive generator.
The triple $(\M, T, \Q)$ is a two-dimensional Borchers triple,
where $T(t_+,t_-) := V(t_+)T^+(t_-)$. It is strictly local if so is $(\M, T^+, \Q)$.
\end{proposition}
\begin{proof}
The Trotter product formula (proved in \cite{Chernoff68} under the assumption here)
tells us that $V(t) = \lim_n \left(V_1(t/n)V_2(t/n)\right)^n$.
Then it is clear that
\[
\Ad V(t)(\M) = \lim_n \Ad \left(V_1(t/n)V_2(t/n)\right)^n(\M)
\subset \bigvee_n \Ad \left(V_1(t/n)V_2(t/n)\right)^n(\M) \subset \M,
\]
since both $\Ad V_1(t/n)$ and $\Ad V_2(t/n)$ are endomorphisms of $\M$.
Analogously $T^+$ commutes with $V$ since so do both $V_1$ and $V_2$.
Hence $V$ implements a one-parameter semigroup of Longo-Witten endomorphisms.
Positivity of the generator $Q_1 + Q_2$ is trivial from the assumptions.

The last statement is just a corollary of Theorem \ref{th:holography}.
\end{proof}
The assumption on the generators could be weakened in order to obtain the
same result, see \cite{Chernoff70, Chernoff74}.

This Proposition indicates that it is important to investigate the set of
one-parameter semigroups of Longo-Witten endomorphisms. Some general properties
have been obtained in \cite{Borchers97}, however, if one aims at constructing
models, it is necessary to study concrete examples.
Even in the best-known case where the Borchers triple comes from
the $\uone$-current net, the known examples of such one-parameter semigroups
are scarce.

There is some hope in models with asymptotic freedom.
Certain integrable models are expected to be asymptotically free \cite{Zinn-Justin93, AAR01},
including the  $\mathrm{O}(N)$ $\sigma$-models treated in \cite{LS12}.
In terms of Algebraic QFT, asymptotic freedom should imply that the scaling limit net
is equivalent to the massless free field net, hence to the tensor product of the
$\uone$-current nets. For an integrable model, the scaling limit net
should be constructed from the one-dimensional Borchers triples as seen in \cite{BLM11},
which is expected to be equivalent to the $\uone$-current net for an asymptotically free models.
Then one conjectures that there are two different one-parameter semigroup of Longo-Witten
endomorphisms, one coming from the free field and the other coming from the interacting field
(they cannot be the same because such a semigroup directly reproduces the net through Theorem \ref{th:holography}).
They could be mixed as Proposition \ref{pr:trotter}, producing further different nets.

\subsubsection*{Trivial examples}
For any Borchers triple $(\M, T, \Q)$ one can take the one-dimensional
reduction $(\M, T^+, \Q)$ and take the two copies of the $t_+$-translation $T^{+\prime}$.
The construction of Proposition \ref{pr:trotter} gives simply
the dilated $T^{+\prime}$. One can take also arbitrarily dilated translation $T^{+\prime}(\k\,\cdot), \k > 0$.
The resulting Borchers triple is just the dilation in the negative lightlike direction.

A slightly more complicated example can be found in \cite{Tanimoto13}.
We take the Borchers triple $(\Mc, \Tc, \Omc)$
which comes from the free massive complex free field.
There is an action of the global gauge group $\uone$, implemented by
$\ee^{\ima 2\pi t\Qc}$. For $\k \in \RR$, one can construct a new Borchers triple
\[
\Mtc := \Mc\otimes \1 \vee \Ad \ee^{\ima 2\pi\k\Qc\otimes\Qc}(\1\otimes\Mc),
\;\;\;\;\;\; \Ttc = \Tc\otimes\Tc,
\;\;\;\;\;\; \Omtc := \Omc\otimes\Omc.
\]
It turned out that this is strictly local.
It is easy to observe that $\Tc^{+\prime}\otimes\1$ and $\1\otimes\Tc^{+\prime}$
are both one-parameter semigroups of Longo-Witten endomorphisms with
the positive generators, since they commute with $\Qc\otimes\Qc$.
Proposition \ref{pr:trotter} changes simply the mass of the left or right
component, correspondingly. A similar observation holds for the construction
in Section \ref{massive}.

With the example above from \cite{Tanimoto13}, it is possible to determine
the lightlike intersection $\Mtc \cap \Ad \Ttc^+(t)(\Mtc)'$.
Indeed, as we know how the commutant looks like, it takes explicitly the form
\[
 \left(\Mc\otimes \1\vee \Ad \ee^{\ima 2\pi\k\Qc\otimes\Qc} \left(\1\otimes\Mc\right)\right)
\cap \left(\Ad \ee^{\ima 2\pi\k\Qc\otimes\Qc} \left(\Mc(t)'\otimes\1\right)\vee \1\otimes\Mc(t)'\right),
\]
where $\Mc(t) = \Ad \Tc^+(t)(\Mc)$.
We have considered this intersection in \cite[Section 4.4]{Tanimoto13} (with the change
of the action from $\ZZ_N$ to $S^1$, which does not affect the argument).
The result is that the above intersection is equal to
$\left(\Mc^\a\cap \Mc^\a(t)'\right)\otimes \left(\Mc^\a\cap\Mc^\a(t)'\right)$,
where $\Mc^\a$ is the fixed point with respect to $\Ad \ee^{\ima 2\pi\a \Qc}$.
The vacuum is clearly not cyclic for this von Neumann algebra if $\a \notin \ZZ$
(if $\a \in \ZZ$, $\ee^{\ima 2\pi\k\Qc\otimes\Qc} = \1$ and this case is not interesting).
In other words, the one-dimensional Borchers triple $(\Mtc, \Ttc^+, \Omtc)$ does not satisfy
strict locality.

\subsubsection*{A non example}
Here we show that on the $\uone$-current net $\u1net$, there is a nontrivial
semigroup of Longo-Witten endomorphisms with positive generator.
The fundamental idea is the boson-fermion correspondence, which we reformulated
in the operator-algebraic approach in \cite[Section 3.3]{BT12}.
In short, the $\uone$-current net can be embedded in the free complex fermion net
$\FerC$, where there is the $\uone$-action by inner symmetry and
it holds that $\u1net = \FerC^{\uone}$, the fixed point subnet.

The net $\FerC$ acts on the fermionic Fock space, where the ``one-particle space''
has actually multiplicity two as the (projective) representation of the M\"obius group
with the lowest weight $\frac12$.
Let us denote by $P_1$ the generator of the translation group on this
``one-particle space''.
The argument of \cite{LW11} works without any essential change for the fermionic
case and one sees that $\Lambda(\ee^{\frac{\ima t}{P_1}})$ implements a Longo-Witten
endomorphism for $t \le 0$, where $\Lambda$ is the fermionic second quantization.
This operator obviously commutes with the inner symmetry and therefore
restricts to the bosonic subspace and implements a Longo-Witten endomorphism
of the $\uone$-current net. The generator is again the restriction and is positive.

One can observe that if the procedure of Theorem \ref{th:holography} is applied to
the $\FerC$ and $\Lambda(\ee^{\frac{it}{P_1}})$, one obtains the massive
free complex fermion net. The proof will be just analogous as in Section \ref{projection}.
This still admits the $\uone$-action. It is immediate that the construction
of Theorem \ref{th:holography}, applied to the $\uone$-current net and this restriction
of the fermionic translation, leads to this $\uone$-fixed point subnet of
the two-dimensional free fermion net.

A natural question arises as to what happens if we mix the two one-parameter semigroups:
one coming from the restriction of the fermionic translation
and the other coming from the bosonic translation, by Proposition \ref{pr:trotter}.
Unfortunately, but interestingly, the self-adjointness condition is crucial.

We show that the common domain does not contain the bosonic one-particle space.
As we calculated in \cite[Section 3.3]{BT12}, the bosonic one-particle space
$L^2(\RR_+, p\dd p)$ can be embedded in the fermionic ``two-particle space''
$\left(L^2(\RR_+, \dd q_+)\oplus L^2(\RR_-,\dd q_-)\right)^{\otimes 2}$ as follows:
For $\Psi \in L^2(\RR_+, p\dd p)$, there corresponds a function
\[
 \iota(\Psi)(q_1,q_2) = -\frac{1}{2\pi}\Psi(q_1-q_2), \;\;\;\;\; \mbox{ for } q_1 > 0, q_2 < 0,
\]
$\iota(\Psi) = 0$ if $q_1$ and $q_2$ have the same sign and on the region
$q_1 < 0 ,q_2 > 0$ it is determined by antisymmetry (note the slight modification of
notation from \cite{BT12}).
The generator of fermionic one-particle translation $P_1$ acts as the multiplication by
$|q|$, hence $\frac{1}{P_1}$ acts by $\frac{1}{|q|}$.
Now we see that any function $\Psi \in L^2(\RR_+, p\dd p)$ is not in the domain of $\frac{1}{P_1}$.
Indeed, we may assume that the support of $\Psi$ contains some $p_0 > 0$.
The multiplication by $\frac{1}{q_1} - \frac{1}{q_2}$ in the fermionic two-particle space gives the function
\[
 \left(\frac{1}{q_1} - \frac{1}{q_2}\right)\iota(\Psi(q_1,q_2))
 = -\frac{1}{2\pi}\left(\frac{1}{q_1} - \frac{1}{q_2}\right)\Psi(q_1-q_2),
\]
which has divergences like $\frac{1}{q_1}$ and $\frac{1}{q_2}$
around $(0,-p_0)$ and $(p_0,0)$, respectively, hence is clearly not in
$L^2(\RR_+, \dd q_1)\otimes L^2(\RR_-, \dd q_2)$.
This implies that $\iota(\Psi)$ is not in the domain of $\Lambda(\frac{1}{P_1})$.

Therefore, we cannot find a common domain in such an elementary way
to apply Proposition \ref{pr:trotter}.
There are still weaker conditions which enable such an addition of
two generators \cite{Chernoff74}, but we are so far not able to check them in this situation.
To the authors' opinion, it is curious that the very existence of Haag-Kastler net is immediately
related to such a domain problem.

\section{Outlook}\label{outlook}
We constructed families of Borchers triples, massless ones with
multiple particle components and nontrivial left-left, right-right and
left-right scatterings and massive ones with block diagonal S-matrix.
Strict locality of these models remains open.
One should note that integrable models with bound states (S-matrix has
poles in the strip \cite{Smirnov92, Quella99}) have not
been treated in the operator-algebraic framework (c.f.\! \cite{LS12}).

We presented also relations between massive models and one-dimensional
Borchers triples accompanied with a one-parameter semigroup of Longo-Witten
endomorphisms with the semibounded generator.
Many open problems in integrable models are relevant with this observation. We discussed
the conjectured relations between the $\mathrm{SU}(2)$-current algebra
and the $\mathrm{SU}(2)$-symmetric Thirring model and asymptotic freedom
in integrable models. We argued that an affirmative solution of any of
these conjectures could lead to further new constructions of strictly local
Borchers triples.

Another correspondence between an integrable model and a conformal field theory
has been recently proposed, e.g.\! \cite{EPSS12}, in connection with higher dimensional gauge theory \cite{AGT10}.
The authors would like to see possible consequences in the operator-algebraic framework.

\subsubsection*{Acknowledgment}
We thank Gerardo Morsella, Roberto Longo, Masaki Oshikawa, Karl-Henning Rehren and L\'aszl\'o Zsid\'o
for interesting discussions.
We are also grateful to Giuseppe Mussardo, Atsushi Nakayashiki and Fedor Smirnov
for answering questions.

\appendix
\section{A Lemma on standard subspaces}
We need a straightforward generalization of a well-known result
(the special case $T=T'$ is basically \cite[Theorem 3.18]{Longo08}).
A general operator $T$ should not be confused with translation.
We use this notation only in this Appendix.

\begin{lemma}[{cf.~\cite[Theorem 3.18]{Longo08}, \cite[Lemma 2.1.]{LW11}}]
	\label{lem:LW}
	Let $K\subset \K$ and $H\subset \H$ be two standard subspaces 
	and $T,T'\in \B(\K,\H)$. Then the following are equivalent:
	\begin{enumerate}
		\item 	\label{it:1}
			$\langle g',Tf\rangle = \langle T'f,g'\rangle$ for all
			$f\in K$, $g'\in H'$.
		\item 	\label{it:e}
			$TS_K\subset S_H T'$.
		\item 	\label{it:f}
			$\Delta_H^{1/2} T \Delta_K^{-1/2}$ is defined  on $\D(\Delta^{-1/2}_K)$ and 
			its closure coincides with $J_HT' J_K$.
		\item \label{it:2}
			The map $T(s):=\Delta_H^{-\ima s} T \Delta_K^{\ima s} \in \B(\K,\H)$,
			$s\in\RR$ extends to a bounded weakly
 			continuous map 
			on $\RR+\ima[0,1/2]$ analytic in $\RR+\ima(0,1/2)$ and satisfying
			$T(\ima/2)=J_HT'J_K$.
	\end{enumerate}	
\end{lemma}
\begin{proof}
	To see \ref{it:e} $\Leftrightarrow$ \ref{it:f} we note that 
	$J_\bullet$ is an involution and $S_\bullet=\Delta_\bullet^{-1/2} J_\bullet$ for $\bullet=H,K$. Therefore,
	$TS_H\subset S_HT'$ is equivalent to $T\Delta_K^{-1/2}\subset \Delta_H^{-1/2} J_H T' J_K$.
	This is equivalent to $\Delta_K^{1/2} T\Delta_H^{-1/2} \xi=J_HT'J_K \xi$ 
	for $\xi \in \D(\Delta_K^{-1/2})$.

	For \ref{it:f} $\Rightarrow$ \ref{it:2} let $\xi\in\H$ and $\eta\in\K$ be entire
	vectors of exponential growth for $\Delta_H$ and $\Delta_K$, respectively.
	We define
	$$
		f_{\xi,\eta}(z):= \langle\xi, \Delta_H^{-\ima z} T\Delta_K^{\ima z} \eta\rangle
			\equiv \langle\Delta_H^{\overline{-\ima z}}\xi, T\Delta_K^{\ima z}\eta\rangle
	$$
 	which is an entire function with $f_{\xi,\eta}(t+\ima/2)= \langle\Delta_H^{1/2}\Delta_H^{\ima t}\xi, T\Delta_K^{-1/2}\Delta^{\ima t}_K\eta \rangle$, which equals 
$\langle \Delta_H^{\ima t} \xi, J_HT'J_K \Delta^{\ima t}_K\eta\rangle$ 
by assuming \ref{it:f}. A priori one has
	the estimate
\[
\|f_{\xi,\eta}(z)\|
        \leq \max \{\|T\|,\|T'\|\} \left\|(\1+ \Delta_H^{-1/2})\xi\right\| \left\|(\1+ \Delta_K^{-1/2})\eta\right\|
\]
        and this can be improved to
        $\|f_{\xi,\eta}(z)\|
        \leq \max \{\|T\|,\|T'\|\} \|\xi\|\, \|\eta\|$ by the three-line theorem (e.g.\! \cite{RSII}).
	By density of $\xi$'s in $H$ and $\eta$'s in $K$, \ref{it:2} follows.
	
	\ref{it:2} $\Rightarrow$ \ref{it:f}: Assuming \ref{it:2} as in the step before 
	$\langle T\Delta_K^{-1/2}\eta,\Delta_H^{1/2}\xi\rangle=\langle J_H T' J_K\eta,\xi\rangle$
	holds and the $\eta$'s and  $\xi$'s form a core for $\Delta_K^{-1/2}$ 
	and $\Delta_H^{-1/2}$, respectively. It follows that the equation holds 
	for all $\eta \in \D(S_K)$ and $\xi \in \D(S_H)$.
	This implies $\Delta_H^{1/2}T\Delta_K^{-1/2}\eta=J_HT'J_K\eta$ for all 
	$\eta \in \D(S_K)$, namely \ref{it:f} holds.

	To see that \ref{it:e} implies \ref{it:1} we calculate for $g'\in H'$ and $f\in K$ 
	\begin{align*} 
		\langle g',Tf\rangle &= \langle g',TS_K f\rangle \\
			&=\langle g',S_H T' f\rangle \\
			&=\langle T'f,S_H^\ast g'\rangle \\
			&=\langle T'f,g'\rangle \,\text.
	\end{align*}

	By assuming $\langle g',Tf \rangle = \langle T'f,g'\rangle$ for all $f\in K$ 
	and $g'\in H'$ we get 
	\begin{align*}
		\langle g',S_H T'f\rangle 
		&= \langle T'f,S_H^\ast g'\rangle
		= \langle T'f,g'\rangle = \<g', Tf\>
	\end{align*}
	and because $H'+\ima H'$ is dense in $\H$ it holds $S_H T'f=Tf$ for all $f\in K$.
	Then we have for $f,g\in K$
	\begin{align*}
		TS_K(f+\ima g) &= Tf-\ima Tg
			= S_H T'f - \ima S_H T'g
			= S_H T'(f+\ima g)
	\end{align*}
	and in particular $T S_K\subset S_H T'$ and we showed \ref{it:1} implies \ref{it:e}.
\end{proof}

\def\cprime{$'$}

\end{document}